\PassOptionsToPackage{hyphens}{url}
\documentclass{article}

\usepackage{microtype} 
\usepackage{graphicx}
\usepackage{booktabs}

\usepackage[T1]{fontenc}
\usepackage{amsmath,amssymb}
\usepackage{amsthm}
\usepackage{subcaption}
\usepackage{enumitem}
\usepackage{bbm}
\usepackage{xspace}
\usepackage{makecell}
\usepackage{xcolor} 

\usepackage{hyperref}
 
\usepackage[accepted]{mlsys2025}
\makeatletter
\renewcommand{\Notice@String}{Preprint.}
\makeatother

\theoremstyle{plain}
\newtheorem{theorem}{Theorem}[section]
\newtheorem{lemma}[theorem]{Lemma}

\theoremstyle{definition}
\newtheorem{definition}[theorem]{Definition}

\newcommand{\E}{\mathbb{E}}

\newcommand{\arrival}{\mathcal{A}}
\newcommand{\policybest}{BOA\xspace}
\newcommand{\NumEpoch}[1]{\ell_{#1}}

\newcommand{\boa}{BOA Constrictor\xspace}
\newcommand{\base}{\mathcal{K}}
\newcommand{\subhead}[1]{\noindent\textbf{#1}}

\mlsystitlerunning{BOA Constrictor: Squeezing Performance out of GPUs in the Cloud}

\begin{document}

\twocolumn[
\mlsystitle{BOA Constrictor: Squeezing Performance out of GPUs in the Cloud via Budget-Optimal Allocation}

\begin{mlsysauthorlist}
\mlsysauthor{Zhouzi Li}{cmu}
\mlsysauthor{Cindy Zhu}{cmu}
\mlsysauthor{Arpan Mukhopadhyay}{warwick}
\mlsysauthor{Mor Harchol-Balter}{cmu}
\mlsysauthor{Benjamin Berg}{unc}
\end{mlsysauthorlist}

\mlsysaffiliation{cmu}{Computer Science Department, Carnegie Mellon University, Pittsburgh, PA, USA}
\mlsysaffiliation{warwick}{Department of Computer Science, University of Warwick, Coventry, United Kingdom}
\mlsysaffiliation{unc}{Department of Computer Science, UNC Chapel Hill, Chapel Hill, NC, USA}

\mlsyscorrespondingauthor{Zhouzi Li}{zhouzil@andrew.cmu.edu}

\mlsyskeywords{GPU scheduling, cloud computing, resource allocation, machine learning systems}


\begin{abstract}
Companies and research labs commonly generate a significant  model training workload. 
Modern training jobs are \emph{elastic}, meaning they can change their level of parallelism as they run.
A scheduler must then decide how many GPUs to allocate to each running job over time.
While prior cluster schedulers focus on the case where the size of the training cluster is \emph{fixed}, many organizations run their training workloads in \emph{variable-sized} clusters rented from cloud providers.
We observe that the scheduler can greatly improve performance by jointly optimizing the size of the cloud-based cluster and the GPU allocations to each individual job.

Scheduling in cloud-based clusters is difficult because jobs receive a diminishing marginal speedup from each additional GPU.
While using a larger cluster can reduce job completion time (JCT) by increasing jobs' GPU allocations, this improvement comes at a cost.
We therefore propose \boa, a cluster scheduler that optimizes the tradeoff between training cost and average JCT.
\boa is based on theoretical results that leverage a stochastic model of training job performance.
We validate these theoretical results via large-scale simulation and real-world deployment on Amazon AWS.
Compared to state-of-the-art schedulers with a given budget, \boa reduces average JCT by up to $1.6\times$ in our implementation and up to $1.8\times$ in large-scale simulations.
\boa also reduces the budget needed to achieve a given average JCT by up to $2\times$.

\end{abstract}

]

\printAffiliationsAndNotice{}

\section{Introduction}
\label{sec:intro}


\begin{figure}[t]
    \centering
    \includegraphics[width=0.9\linewidth]{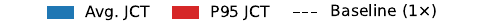}\\[1pt]
    \begin{subfigure}[t]{0.49\linewidth}
        \centering
        \includegraphics[width=\linewidth]{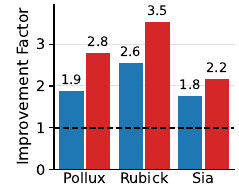}
        \caption{JCT reduction}
        \label{fig:headline:jct}
    \end{subfigure}
    \hfill
    \begin{subfigure}[t]{0.49\linewidth}
        \centering
        \includegraphics[width=\linewidth]{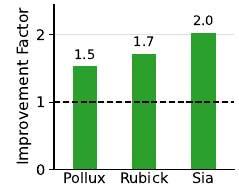}
        \caption{Budget reduction}
        \label{fig:headline:cost}
    \end{subfigure}
    \caption{\boa versus state-of-the-art schedulers Pollux~\cite{qiao2021pollux}, Rubick~\cite{zhang2025rubick}, and Sia~\cite{subramanya2023sia}. (a) At the same budget, \boa reduces average and P95 job completion time (JCT) by up to the factor shown. (b) \boa reduces the budget needed to reach the same average JCT by up to the factor shown.}
    \label{fig:headline}
\end{figure}

Machine Learning (ML) models are increasingly trained on \emph{flexibly-sized clusters} of GPUs rented from cloud providers rather than on dedicated clusters~\cite{jouppi2017datacenter,gpureport}.
A \emph{cloud customer} (e.g., a research lab or a company) submits a stream of training jobs and pays a fixed price for every GPU-hour they use~\cite{salvaris2018microsoft}.
Modern cloud providers let the customer acquire or release GPUs within 1--2 minutes, so the customer can flexibly resize their cluster as jobs arrive and complete.
The customer then faces two decisions at every moment: (i) How many GPUs to rent? and (ii) How to allocate these GPUs to the jobs in the system?

In practice, these decisions are made under a budget.
Organizations typically manage cloud spending against a fixed monthly or yearly \emph{budget}~\cite{awsbudgets}, and GPUs are a fast-growing part of this spending: organizations that use cloud GPUs increased their spending on them by 40\% in a single year~\cite{datadog2024cloudcosts}.

Managing this spending well, however, is hard: 84\% of organizations name managing cloud spend as their top cloud challenge~\cite{flexera2025cloud}.
In particular, a budget does not tell the customer how to spend it.
Because training jobs arrive in bursts, the customer must repeatedly choose between spending more now to finish jobs sooner and spending less to keep budget for later.
Existing budget tools only send alerts once spending crosses a threshold~\cite{awsbudgets}, so this choice is left to the customer. 

Renting more GPUs, or giving more GPUs to a job, generally reduces \emph{job completion time} (JCT) but increases the cost of training.
Hence, given a budget, the customer must jointly optimize the cluster size and job allocations over time to optimally balance the \emph{cost-performance tradeoff}. 
We refer to this problem as the \emph{GPU rental problem}.

\begin{figure}[t]
    \centering
    \begin{subfigure}{0.48\linewidth}
        \centering
        \includegraphics[width=0.9\linewidth]{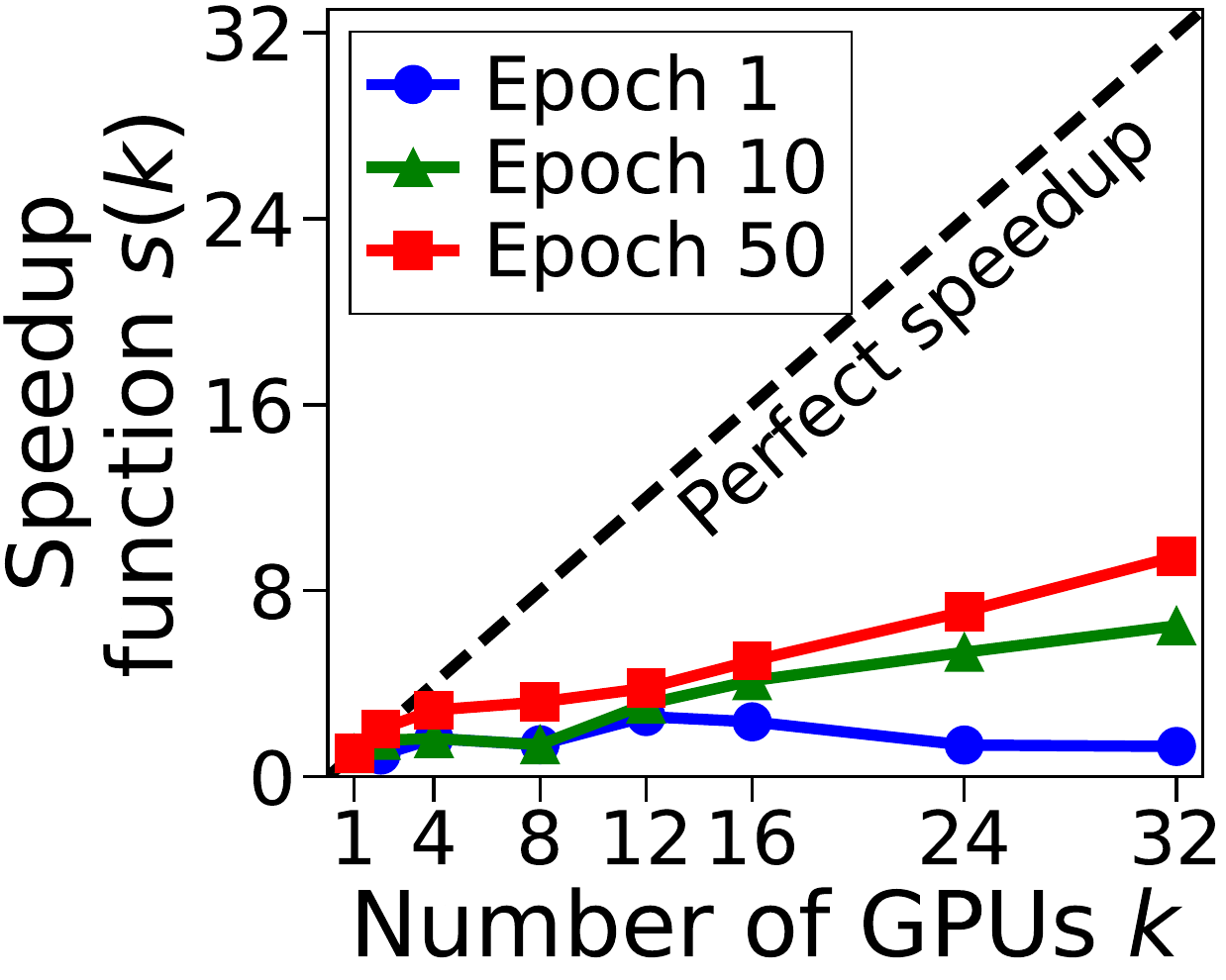}
        \caption{Speedup functions}
    \label{fig:speedup:changing}
    \end{subfigure}
    \hfill
    \begin{subfigure}{0.48\linewidth}
        \centering
        \includegraphics[width=0.9\linewidth]{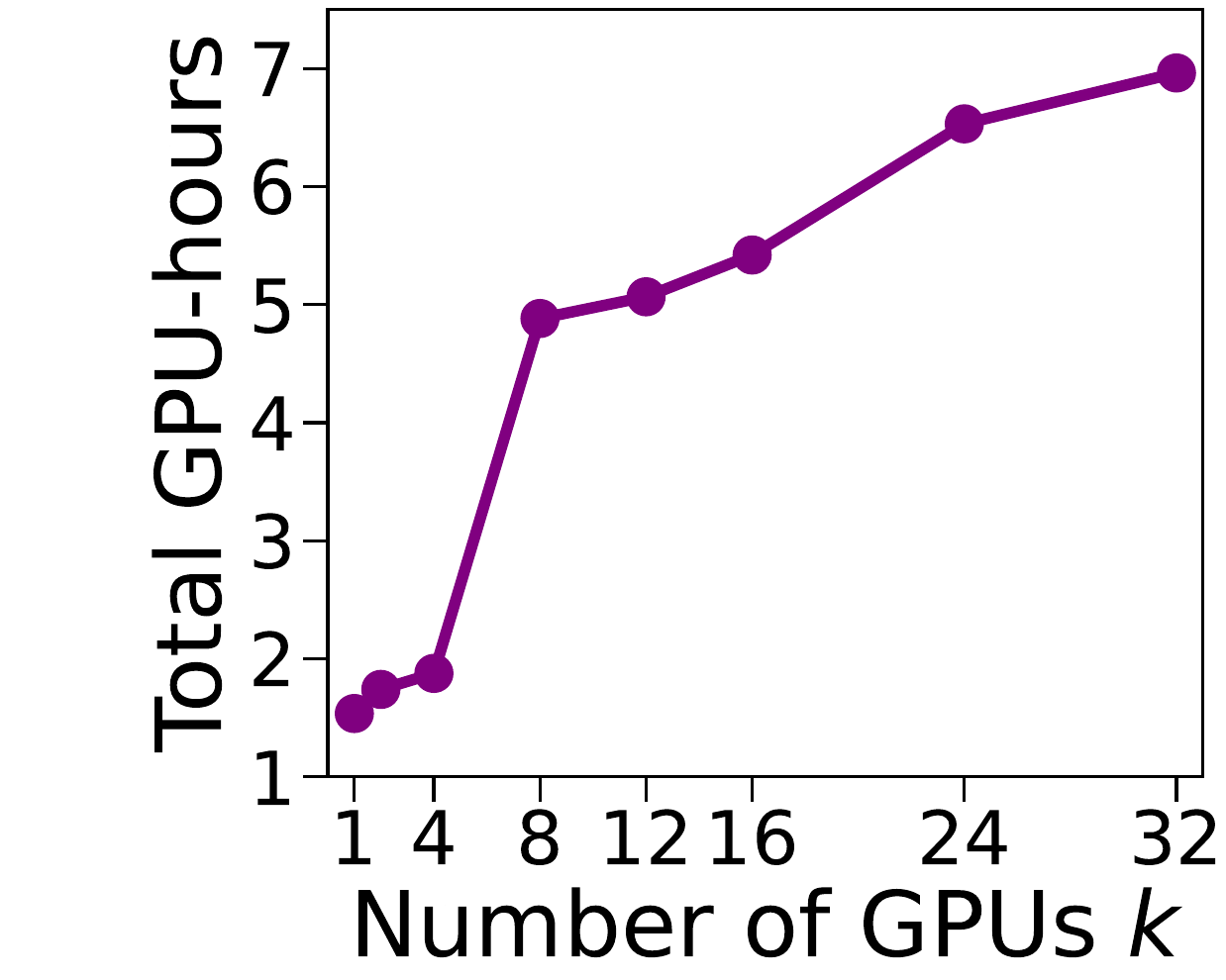}
        \caption{Training cost}
    \label{fig:usage}
    \end{subfigure}
    \caption{{A CIFAR-10 training job on different numbers of GPUs. Following \cite{qiao2021pollux,subramanya2023sia,zhang2025rubick}, we divide training into \emph{statistical epochs} and optimize job hyperparameters at each epoch.  This causes the job's speedup function to change over time.  Because the speedup is sublinear, training cost rises with the number of GPUs.}}
    \vspace{-.25in}
\end{figure}

\subhead{Why the GPU rental problem is hard.}
Some training cluster schedulers \cite{moritz2018ray} have historically assumed  that each training job runs on a predetermined number of GPUs.
However, modern training jobs are \emph{elastic}, meaning a job can run on any number of GPUs and it can change its degree of parallelism over time.
Hence, it is unclear how many total GPUs should be rented at any given time to serve the set of elastic jobs currently in the system.
In particular, sequential bottlenecks, synchronization overheads, and statistical inefficiency (Section~\ref{sec:bg:train}) make a job's \emph{speedup function} $s(k)$ sublinear and concave (Figure~\ref{fig:speedup:changing}), where $s(k)$ is the factor by which the job's training progress rate on $k$ GPUs exceeds its rate on a single GPU. 
While it is tempting to rent a sufficient number of GPUs to give large allocations to all jobs in the system, the concavity of job speedup functions can make this decision highly inefficient.
Specifically, a job run on $k$ GPUs will finish in $\frac{1}{s(k)}$ of the time but will consume $\frac{k}{s(k)}$ times as many GPU-hours.
Therefore, giving large allocations to individual jobs reduces JCT (Figure~\ref{fig:speedup:changing}) at the expense of increased resource usage (Figure~\ref{fig:usage}).

To make matters worse, jobs differ in parallelizability, and a job's speedup function can change during training as its optimal hyperparameters change (Figure~\ref{fig:speedup:changing})~\cite{qiao2021pollux}.
Models also differ in training characteristics, and the workload composition changes over time (Section~\ref{sec:bg:workload}).
Consequently, both the optimal cluster size and per-job GPU allocations change over time.

\subhead{Existing schedulers ignore half of the problem.}
State-of-the-art schedulers for ML training, such as Pollux~\cite{qiao2021pollux}, Sia~\cite{subramanya2023sia}, and Rubick~\cite{zhang2025rubick}, are designed for \emph{fixed-size} clusters (Section~\ref{sec:bg:prior}).
The only way to control these schedulers' training costs is by choosing a fixed cluster size offline.
Unfortunately, any fixed-size cluster will have wasteful periods of underutilization, and periods of overload that lead to high JCTs.
Existing schedulers attempt to mitigate these periods of under and over utilization by continually reshuffling the fixed set of GPUs across jobs.
However, any scheduler that uses a fixed-size cluster has an inherent disadvantage in balancing the cost-performance tradeoff.

Conversely, existing work on autoscaling for ML training jobs only targets a single job~\cite{tyagi2023scavenger} and does not extend to a stream of jobs.
We show that na\"ively applying these single-job methods to an entire training workload does not adequately balance the cost-performance tradeoff \emph{across jobs}.
To our knowledge, no existing system jointly optimizes the cluster size and each job's degree of parallelism for a stream of jobs.

\subhead{Our solution: \boa.}
\boa takes a budget from the customer (e.g., \$10{,}000 per month) and makes the rental and allocation decisions jointly to minimize average JCT, while guaranteeing that long-run average spending stays within the budget.
Compared with the state-of-the-art schedulers Pollux, Rubick, and Sia, \boa reduces average JCT by up to $1.8$--$2.6\times$ and P95 JCT by up to $2.2$--$3.5\times$ at the same budget, and it reaches the same average JCT with up to $2\times$ less budget (Figure~\ref{fig:headline}).

The core of \boa is the \emph{Budget-Optimal Allocation} (BOA) policy.
We formalize the GPU rental problem as a budget-constrained stochastic scheduling problem  (Section~\ref{sec:model:problem}) and use our model to derive the BOA policy (Section \ref{sec:BOA}).
Our derivation proves two somewhat counterintuitive structural properties of BOA:

\noindent\textbf{1.} The optimal policy should not intentionally queue jobs.

\noindent\textbf{2.} The optimal policy is \emph{fixed-width}, meaning the number of GPUs a job receives depends only on the job's speedup function and the budget, not on the other jobs in service.

From these properties, we see that a job should be allocated GPUs the moment it arrives and rescaled only when its own speedup function changes.
The cluster size should then grow and shrink to accommodate these allocations.
BOA computes these allocations via an optimization problem that considers overall workload composition, reserving budget for the jobs where it buys the greatest JCT reduction per GPU-hour. Notably,
BOA only uses information available to existing schedulers: estimates of job speedup functions and coarse per-class workload statistics (job arrival rates and mean training work per job). 

Because the BOA policy can be computed for any budget level, \boa exposes the entire cost-performance Pareto frontier (See Figure~\ref{fig:sim}). 
This allows the customer to pick their preferred Pareto-optimal operating point instead of tuning the scheduler by hand.


Finally, \boa is orthogonal to hyperparameter tuning: it consistently improves performance when layered on Pollux, Rubick, and Sia, which use different hyperparameter optimizations (Figure~\ref{fig:headline} and Section~\ref{sec:eval}).

\subhead{Contributions.}
\begin{itemize}[leftmargin=*,topsep=2pt,itemsep=1pt,parsep=0pt]
\item We formalize the GPU rental problem, in which a cloud customer minimizes the average JCT of a stream of training jobs subject to a long-run average budget (Section~\ref{sec:model:problem}), and derive the optimal policy, BOA (Section~\ref{sec:BOA}).
Unlike prior theory work, which studies either cluster scaling or GPU allocation alone, we optimize both jointly.
\item We implement \boa on top of the AdaptDL framework~\cite{adaptdl} and deploy it on a variable-sized GPU cluster on AWS (Section~\ref{sec:design}). 
\item We evaluate \boa against state-of-the-art schedulers in large-scale simulations and on a physical AWS cluster (Section~\ref{sec:eval}).
\boa substantially reduces JCT at the same budget and reaches the same JCT with less budget (Figure~\ref{fig:headline}), and these gains carry over to the physical cluster.
\end{itemize}

\section{Background}
\label{sec:bg}
{This section covers how distributed training jobs run in the cloud (Section~\ref{sec:bg:train}), the workload properties that make the GPU rental problem dynamic (Section~\ref{sec:bg:workload}), and why existing schedulers do not solve it (Section~\ref{sec:bg:prior}).}

\subsection{Distributed Training in the Cloud}\label{sec:bg:train}
\subhead{Parallelism and the speedup function.}
A training job is distributed across GPUs either by \emph{data parallelism}, which replicates the model and splits each batch across the replicas, or by \emph{model parallelism}, which splits the model itself across GPUs (pipeline or tensor parallelism).
We mainly focus on data-parallel jobs, which eases comparison with prior work, and consider other modes of parallelism in Section~\ref{sec:eval:setup}.

A job is complete when it reaches a target training loss or validation accuracy.
Two overheads make the speedup from data parallelism sublinear~\cite{qiao2021pollux}.
First, replicas must synchronize gradients after every step, via all-reduce or a parameter server, which limits \emph{throughput}, the number of samples processed per second. 
Second, running on more GPUs requires a larger global batch, which reduces \emph{statistical efficiency}, the progress made per sample.
The product of throughput and statistical efficiency measures how fast a job makes progress on $k$ GPUs, and hence its speedup function $s(k)$ (Figure~\ref{fig:speedup:changing}).
Other modes of parallelism incur analogous overheads and yield similarly sublinear speedups~\cite{zhang2025rubick}.

\subhead{Hyperparameter selection.}
A job's speedup depends on hyperparameters including learning rate and batch size~\cite{qiao2021pollux}, parallelism mode~\cite{zhang2025rubick}, and model architecture.
Because some optimal hyperparameters, notably batch size, depend on GPU count, several recent systems couple hyperparameter optimization with GPU allocation~\cite{qiao2021pollux,subramanya2023sia,zhang2025rubick,misra2021rubberband,dunlap2021elastic}.
These systems pair sophisticated hyperparameter search with fairly simple scheduling heuristics, suggesting substantial room for better scheduling.

\subhead{Our approach.}
We decouple scheduling from hyperparameter search, defining a job's speedup on $k$ GPUs using the best known hyperparameters for that GPU count.
The scheduler estimates this function online (e.g., through AdaptDL's goodput model~\cite{qiao2021pollux}).
\boa uses these estimates without separate profiling, and our results complement prior work on hyperparameter tuning.

\subhead{Training in the cloud.}
Cloud providers rent GPU instances on demand and bill per GPU-hour, with different prices for different GPU types.
Instances can be acquired or released within minutes; in our AWS deployment (Section~\ref{sec:design:overheads}), a newly rented node takes about two minutes to become useful for training, and we account for this time as part of rescaling overhead.
We consider on-demand instances.
Exploiting cheaper, revocable spot instances~\cite{thorpe2023bamboo,duan2024parcae} is orthogonal to our results.

\subsection{Workloads}\label{sec:bg:workload}
The GPU rental problem is dynamic because training workloads vary along three dimensions, which we observe in the production traces used in Section~\ref{sec:eval:setup}~\cite{qiao2021pollux,subramanya2023sia}.
\subhead{Job running times are highly variable.}
On a fixed number of GPUs, the time to complete a job varies by more than an order of magnitude depending on the model and the target accuracy.
\subhead{Job arrivals are bursty.}
The arrival rate fluctuates over time, so a fixed cluster is alternately idle and overloaded; cluster autoscaling can absorb this variability~\cite{gandhi2012autoscale}.
\subhead{Speedup functions change over time.}
Statistical efficiency tends to be low early in training and to rise later~\cite{qiao2021pollux}, so a job's speedup function shifts upward as training progresses, and the right allocation for the job shifts with it.
{A rental policy must respond to all three, and this variability is exactly what a fixed-size cluster cannot follow.}

\subsection{Prior Work: Current Schedulers Fall Short}\label{sec:bg:prior}
\subhead{Prior theoretical work.}
The scheduling-theory literature has studied allocating a fixed number of servers among parallelizable jobs with speedup functions~\cite{berg2018,berg2020hesrpt,berg2022case,berg2020optimal} and, separately, autoscaling a pool of servers to serve a stream of jobs~\cite{gandhi2012autoscale,gandhi2011distributed,psychas2022theory}.

Prior systems fall into three categories.

\subhead{Approach 1: Reservation-based systems.}
Schedulers such as Ray~\cite{moritz2018ray}, Tiresias~\cite{gu2019tiresias}, and AntMan~\cite{xiao2020antman} have the customer specify the number of GPUs each job needs. These demands are then fulfilled using heuristics such as data locality (Ray) or the Gittins index (Tiresias).
Cost-aware cloud schedulers~\cite{chung2018stratus,chang2025eva} additionally choose which cloud instances to rent to meet the job demands at minimum cost.
{All of these systems leave the cost-performance tradeoff to the customer, who chooses each job's allocation without guidance and cannot adapt this choice to changes in the job or the workload.}

\subhead{Approach 2: Elastic jobs in a fixed-size cluster.}
Significant prior work considers \emph{elastic} training jobs where the scheduler sets each job's GPU allocation to optimize objectives like average JCT or goodput~\cite{qiao2021pollux,subramanya2023sia,zhang2025rubick,peng2018optimus,hu2023lucid,narayanan2020heterogeneity,le2020allox}, deadlines or makespan~\cite{gu2023elasticflow,zheng2023shockwave,li2023easyscale}, or fairness~\cite{mahajan2020themis}.
All these systems assume a \emph{fixed-size cluster}.
Aryl~\cite{li2022aryl} and Lyra~\cite{li2023lyra} grow the training pool by borrowing idle inference GPUs, but still fundamentally consider a fixed-size cluster.
Among those that optimize JCT, Pollux, Sia, and Rubick represent the state of the art.
We compare against these three systems in Section \ref{sec:eval}.

\subhead{Na\"ive autoscaling.}
Pollux~\cite{qiao2021pollux} proposes a heuristic for adapting a fixed-size cluster schedule to the cloud setting. 
The proposed scheduler tracks the cluster-wide training efficiency and rents or releases GPUs to keep it near a target value (see Section~\ref{sec:eval:autoscaling} for details).
For the sake of comparison, we implement versions of this proposal on top of Pollux, Sia, and Rubick, showing that this na\"ive autoscaling is not sufficient to balance the cost-performance tradeoff (Section \ref{sec:eval:sim}).

\subhead{Approach 3: Autoscaling for a single job.}
Work that does scale the cluster to balance cost and performance targets a single training job~\cite{tyagi2023scavenger}, or the trials of a single hyperparameter-tuning job with a deadline and a monetary budget~\cite{dunlap2021elastic}.
None of these extends to a stream of independent jobs.

\section{Our Model}
\label{sec:model:problem}

We model the GPU rental problem for a \emph{customer} who submits a stream of ML training jobs to the cloud and pays for every GPU they rent.
To keep notation light, we describe the model for a homogeneous cluster and add the notation for heterogeneous clusters in Section~\ref{sec:hetero}.
The model needs only a few long-run averages of the workload to be well defined; we list our assumptions at the end of the section.

\subhead{Training jobs.}
We view each training job as carrying some \emph{inherent work}, the statistical progress the model must make to reach its target accuracy.
Jobs belong to one of $M$ \emph{classes}, each a combination of model and training data.
The inherent work of a class-$i$ job is a random variable $X_i$; the inherent work of different jobs need not be independent.
 
A job's \emph{speedup function}, $s(k)$, is the rate at which it processes inherent work on $k$ GPUs after its hyperparameters and placement have been optimized.
BOA chooses the allocation; the underlying scheduler handles hyperparameter selection and placement (Section~\ref{sec:design}).

A job's speedup function changes as training progresses.
Following~\cite{qiao2021pollux}, we divide each class-$i$ job into $\NumEpoch{i}$ \emph{statistical epochs}.
Epoch $j\in\{0,\ldots,\NumEpoch{i}-1\}$ has inherent work $X_{ij}$ with finite mean $\E[X_{ij}]$ and speedup function $s_{ij}(k)$, so that $X_i=\sum_j X_{ij}$, and a class-$i$ job in epoch $j$ on $k$ GPUs completes the epoch in time $X_{ij}/s_{ij}(k)$.

Each class, $i$, has a \emph{minimum allocation} $\base_i$, the fewest GPUs on which a class-$i$ job can run, and $s_{ij}(k)$ is defined for $k\ge\base_i$.
This captures jobs that do not fit on few GPUs, for example $\base_i=8$ for fine-tuning a large language model.
For the analysis we allow allocations to be fractional; prior work has shown that
fractional allocations can be rounded to integers with little loss~\cite{10.1145/3641512.3686370}.

\subhead{Arrival process and load.}
Jobs arrive over time at a long-run average rate of $\lambda$ jobs per second, of which class-$i$ jobs arrive at rate $\lambda_i$.
The product $\rho_{ij}=\lambda_i\,\E[X_{ij}]$ is the long-run rate at which epoch-$j$ work of class $i$ enters the system; we call it the \emph{load} of that epoch.

\subhead{Rescaling.}
Changing a job's GPU allocation incurs a \emph{rescaling overhead}: the job is checkpointed, its runtime environment, container image, and training data are set up on the new GPUs, and training resumes.
During this time the job holds GPUs but makes no progress.
We denote the rescaling time of a class-$i$ job by a random variable $R_i$ with mean $r_i$.

\subhead{Cost-performance tradeoff.}
An \emph{allocation policy} decides, at every moment, how many GPUs to rent and how many to give to each job.
We measure \emph{performance} by job completion time (JCT), the time from a job's arrival until it completes, and optimize the \emph{average JCT}, $\E[T]$, over the stream of arriving jobs.
Our results extend directly to a weighted average JCT, where the weight of each class captures its importance or a fairness objective.
We do not optimize tail latency explicitly, but Section~\ref{sec:eval} shows that \boa also improves P95 JCT.

We measure \emph{cost} in dollars.
With homogeneous GPUs, every GPU costs the same per hour, so cost is equivalent to GPU-hours; with heterogeneous GPUs (Section~\ref{sec:hetero}), each GPU type has its own hourly price.
We write $\bar K$ for the \emph{time-average cost} of a policy, i.e., its long-run spending rate.
Given a budget $b$ on this rate, the customer's problem is to find the policy that minimizes $\E[T]$ subject to $\bar K\le b$.

\subhead{Modeling assumptions.}
We make three assumptions.

First, the interarrival times, the inherent work of each epoch, and the rescaling times each have a well-defined long-run average, and the customer knows these averages.
Interarrival times may be correlated, so arrivals may be bursty, and the inherent work of different jobs may be correlated.
In practice, the customer estimates the averages from past jobs and refines them as new jobs run; Section~\ref{sec:eval:sensitivity} shows that \boa performs well from such estimates.

Second, each speedup function $s_{ij}$ is continuous, non-decreasing, and concave on $[\base_i,\infty)$.
Measured speedup functions such as those in Figure~\ref{fig:speedup:changing} can violate monotonicity or concavity locally.
Prior work~\cite{10.1145/3641512.3686370} shows that a job can match the concave envelope of its speedup function. Hence we replace every measured speedup function by its monotone concave envelope.

Third, the cloud provider can always supply the GPUs the customer requests.
In practice, providers enforce instance limits capping each customer's usage.
However, these limits are typically significantly higher than the average usage and they never affected our experiments in Section \ref{sec:eval}.


\section{Budget-Optimal Allocation}
\label{sec:BOA}

BOA allocates GPUs to minimize the (weighted) average JCT while keeping the time-average cost under the budget $b$.
We first build intuition in the simplest setting, a homogeneous cluster without rescaling overheads.

Prior schedulers rely on various heuristics: maximize the total speedup across running jobs~\cite{qiao2021pollux,subramanya2023sia}, a form of water-filling~\cite{tse2005fundamentals}; rent or release GPUs to hold cluster efficiency at a target~\cite{qiao2021pollux}; and, in the scheduling-theory literature, finish short jobs first~\cite{berg2020hesrpt} or prioritize the less parallelizable jobs~\cite{berg2024asymptotically}.
None of these is optimal under a budget on time-average spending.
Instead, we prove that the optimal policy has two structural properties that are counter-intuitive from the fixed-cluster point of view:

\begin{lemma}[No queueing; informal version of Lemma~\ref{lemma:no queue}]
\label{lemma:no queue:informal}
Under the optimal policy, every job is allocated GPUs as soon as it arrives; no job ever waits in a queue.
\end{lemma}
{Queueing a job lowers the number of GPUs in use at that instant, but the job must still be processed later, so it does not lower the long-run average that the customer pays for; queueing therefore worsens JCT without saving money.}

\begin{lemma}[Fixed width; informal version of Lemma~\ref{lemma:fw}]
\label{lemma:fw:informal}
Under the optimal policy, a job's GPU allocation depends only on its class and current epoch: every class-$i$ job in epoch $j$ runs on $k_{ij}$ GPUs throughout the epoch.
\end{lemma}
Because speedup functions are concave, using a fixed number of GPUs to process a job is faster on average than alternating between high and low allocations.
The same argument applies across jobs of the same class and epoch.

Proven in Appendix~\ref{app:offline}, these lemmas characterize the optimal policy as fixed-width: the cluster grows and shrinks to accommodate the prescribed job allocations.
 
\begin{definition}
\label{def:fixedwidth}
A \emph{fixed-width} policy is parameterized by a set of constants $\{k_{ij}\}$: It assigns $k_{ij}$ GPUs to every type-$i$ job for the entire duration of epoch $j$, without queueing.
\end{definition}

\begin{definition}[BOA policy]
\label{def:boa}
For any budget $b$, BOA is the fixed-width policy whose widths $\{k_{ij}\}$ minimize (weighted) average JCT subject to the time-average budget $b$.
\end{definition}

When rescaling is free, the best fixed-width policy, BOA, is optimal among all policies, in homogeneous (Section~\ref{sec:theory:hmo}) and heterogeneous (Section~\ref{sec:hetero}) clusters.
Section~\ref{sec:theory:rescaling} then accounts for rescaling overheads.

\subsection{BOA in Homogeneous Clusters}
\label{sec:theory:hmo}
 
\begin{theorem}[Optimality of BOA]
\label{thm:fw-optimal}
In the homogeneous setting without rescaling overheads, the budget-optimal allocation (BOA) policy is optimal and can be computed by solving the optimization problem:
\begin{equation*}
\begin{aligned}
& \underset{\{k_{ij}\}}{\min}
& & \textit{average JCT} = \frac{1}{\lambda} \sum_{i} \lambda_i \sum_{j} \frac{\E[X_{ij}]}{s_{ij}(k_{ij})}  \\
& \mbox{subject to}
& & \textit{budget}=\sum_{i} \lambda_i \sum_{j} k_{ij}\,\frac{\E[X_{ij}]}{s_{ij}(k_{ij})} \le b, \\
& & & k_{ij}\ge \base_i.
\end{aligned}
\end{equation*} 

A change of variables makes it a convex optimization.
\end{theorem}

\subhead{Where the two expressions come from.}
Under a fixed-width policy, every class-$i$ job in epoch $j$ runs on $k_{ij}$ GPUs and completes the epoch in expected time $\E[X_{ij}]/s_{ij}(k_{ij})$.
Summing over a job's epochs gives its expected JCT, and weighting each class by its share $\lambda_i/\lambda$ of arrivals gives the average JCT in the objective.
Throughout the epoch the job holds $k_{ij}$ GPUs, so class-$i$ jobs in epoch $j$ consume GPU-hours at rate $\lambda_i\,k_{ij}\,\E[X_{ij}]/s_{ij}(k_{ij})$; summing over classes and epochs gives the operating budget in the constraint.

Appendix~\ref{app:offline} proves that no policy achieves a lower average JCT within the budget (Theorem~\ref{thm:hmo:summary}), and Appendix~\ref{sec:appendix convex} gives the change of variables that makes the problem convex, so the widths can be computed with a standard solver.

\subsection{BOA in Heterogeneous Clusters}
\label{sec:hetero}

We now allow $H$ types of GPUs.
The customer can rent GPUs of any type, but each type has its own performance and price: $s^{(h)}_{ij}(k)$ is the rate of a class-$i$ job in epoch $j$ on $k$ type-$h$ GPUs, $c^{(h)}$ is the hourly price of one type-$h$ GPU, and $\base^{(h)}_i$ is the minimum allocation of class $i$ on type $h$.
Following~\cite{subramanya2023sia}, a job runs on a single GPU type at a time.

The optimal policy jointly chooses each job epoch's GPU type and allocation width under a single budget constraint.
BOA represents these choices by $p^{(h)}_{ij}$, the fraction of class-$i$ jobs running epoch $j$ on type $h$, and widths $k^{(h)}_{ij}$.

\begin{theorem}[BOA in heterogeneous clusters]
\label{thm:hetero-fw-optimal}
In the heterogeneous setting without rescaling overheads, the optimal policy is fixed-width within each GPU type, and its routing fractions $\{p^{(h)}_{ij}\}$ and widths $\{k^{(h)}_{ij}\}$ can be computed by solving a convex optimization problem.
\end{theorem}
The optimization problem is similar to that of Theorem~\ref{thm:fw-optimal}; it is written out as~\eqref{eq:hetero-opt-norescale} in Appendix~\ref{app:simplification}.
The proof is that of Appendix~\ref{app:offline} applied separately to the jobs on each GPU type.
To implement the policy, each epoch is routed to type $h$ at random with probability $p^{(h)}_{ij}$.

\subsection{Adding Rescaling Costs to BOA}
\label{sec:theory:rescaling}

{In practice, changing a job's allocation is not free (Section~\ref{sec:design:overheads}): a rescaling job holds GPUs but makes no progress, which raises both its JCT and the customer's cost. A fixed-width policy rescales a job whenever it enters an epoch with a different width.}
To add these rescaling costs to our model, we introduce an indicator $\mathbf{1}_{ij}$ (resp. $\mathbf{1}^{(h)}_{ij}$) to our optimization problem for the homogeneous (resp. heterogeneous) setting.
This indicator is set to $1$ if a job changes its allocation between epochs $j-1$ and $j$.
Each rescaling adds $r_i$ to the time the job spends in epoch $j$, during which it holds $k_{ij}$ (resp. $k^{(h)}_{ij}$) GPUs.
For example, when a job rescales in the heterogeneous case, the epoch's expected duration becomes $\E[X_{ij}]/s^{(h)}_{ij}(k^{(h)}_{ij})+r_i\mathbf{1}^{(h)}_{ij}$.

To limit rescaling costs, we introduce one notable simplification to the optimization problem for the heterogeneous setting:
Each arriving class-$i$ job is routed to type $h$ with probability $p^{(h)}_i$ and stays there for its lifetime, so the optimization chooses per-class fractions $p^{(h)}_i$ instead of per-epoch fractions.
This makes our optimization problem tractable and, in practice, made no difference in performance (Appendix~\ref{app:simplification}).

Definition \ref{def:BOA-hetero-rescale} gives the optimization problem for the heterogeneous case with rescaling overheads. This is a generalization of the homogeneous case with rescaling overheads.

\begin{definition}[Heterogeneous BOA w/ rescaling]
\label{def:BOA-hetero-rescale}
{Each arriving class-$i$ job is routed to GPU type $h$ with probability $p_i^{(h)}$ and stays there for its lifetime; within each type, jobs follow a fixed-width policy with widths $k_{ij}^{(h)}$ that solve:}
\begin{equation}
\begin{aligned}
& \underset{\{p_i^{(h)},k_{ij}^{(h)}\}}{\min}
& & \textit{average JCT} \\
& \mbox{subject to}
& & \textit{budget}\le b, \\
& & & \sum_{h\in H}p_i^{(h)}=1\quad\forall i, \\
& & & k_{ij}^{(h)}\ge \base_i^{(h)},\quad p_i^{(h)}\in[0,1],
\end{aligned}
\label{eq:hetero-opt-rescale}
\end{equation}
where
\begin{align*}
&\mbox{average JCT} = \frac{1}{\lambda}\sum_{i,h}\lambda_i\,p_i^{(h)}\sum_j\Big(\frac{\E[X_{ij}]}{s_{ij}^{(h)}(k_{ij}^{(h)})}+r_i\mathbf{1}_{ij}^{(h)}\Big), \\
&\mbox{budget} = \sum_{i,h}c^{(h)}\lambda_i\,p_i^{(h)}\sum_j k_{ij}^{(h)}\Big(\frac{\E[X_{ij}]}{s_{ij}^{(h)}(k_{ij}^{(h)})}+r_i\mathbf{1}_{ij}^{(h)}\Big),
\end{align*}
and $\mathbf{1}_{ij}^{(h)}=1$ if $k_{ij}^{(h)}\neq k_{i(j-1)}^{(h)}$ or $j=0$, and $0$ otherwise.
\end{definition}

{
The addition of indicators makes Problem~\eqref{eq:hetero-opt-rescale} a mixed-integer convex program (MICP).
However, we show that it is tractable in practice (Section~\ref{sec:design:overheads}).
While a fixed-width policy may not be optimal given rescaling overheads, BOA is still the best fixed-width policy.
BOA is the core of \boa, which Sections~\ref{sec:design} and~\ref{sec:eval} describe and evaluate.}

\section{System Design}
\label{sec:design}

In this section, we describe \boa, our scheduler that realizes the \policybest policy.

\subsection{Implementing the BOA Policy}
\label{sec:implementing_boa}
 
We implemented \boa on AdaptDL, the open-source scheduling framework that implements Pollux~\cite{qiao2021pollux,adaptdl}.
AdaptDL schedules and rescales training jobs according to a pluggable scheduling policy, and its control pods learn each job's speedup function by profiling the job as it runs (Figure~\ref{fig:fw}); \boa (marked in red) consumes these learned speedup functions.

The \emph{\boa Width Calculator} solves MICP~\eqref{eq:hetero-opt-rescale} in the background using the profiler's latest speedup functions and workload statistics, refreshing the \emph{width table} $\{k^{(h)}_{ij}\}$ every $15$ minutes in our deployment (Figure~\ref{fig:fw}).
At each scheduling round, the synchronous scheduler reads and executes this table without solving an optimization problem.
Separating width computation from execution keeps optimization off the critical path and incorporates refined speedup estimates without blocking jobs.

\paragraph{Synchronous fixed-width execution.}

After each refresh, the updated width table is applied to all active jobs at the next scheduling round.
At each scheduling round, the scheduler consults the table to make two decisions:

{\noindent \textbf{1. Job allocation.} For each GPU type $h$, class-$i$ jobs in epoch $j$ are assigned $k^{(h)}_{ij}$ GPUs.
    \emph{Which} GPUs they get is decided by a greedy placement heuristic similar to Sia's~\cite{subramanya2023sia}: jobs already running on the right number of GPUs keep their allocations, and each job's GPUs are packed onto as few nodes as possible.} 
    
\noindent\textbf{2. Cluster sizing.} The desired cluster size is the sum of $k^{(h)}_{ij}$ across all active jobs. This sum is converted to a number of nodes demanded and forwarded to the AdaptDL Cluster Expander, which provisions or releases cloud instances.



\begin{figure}[h]
    \centering
    \includegraphics[width=0.95\linewidth]{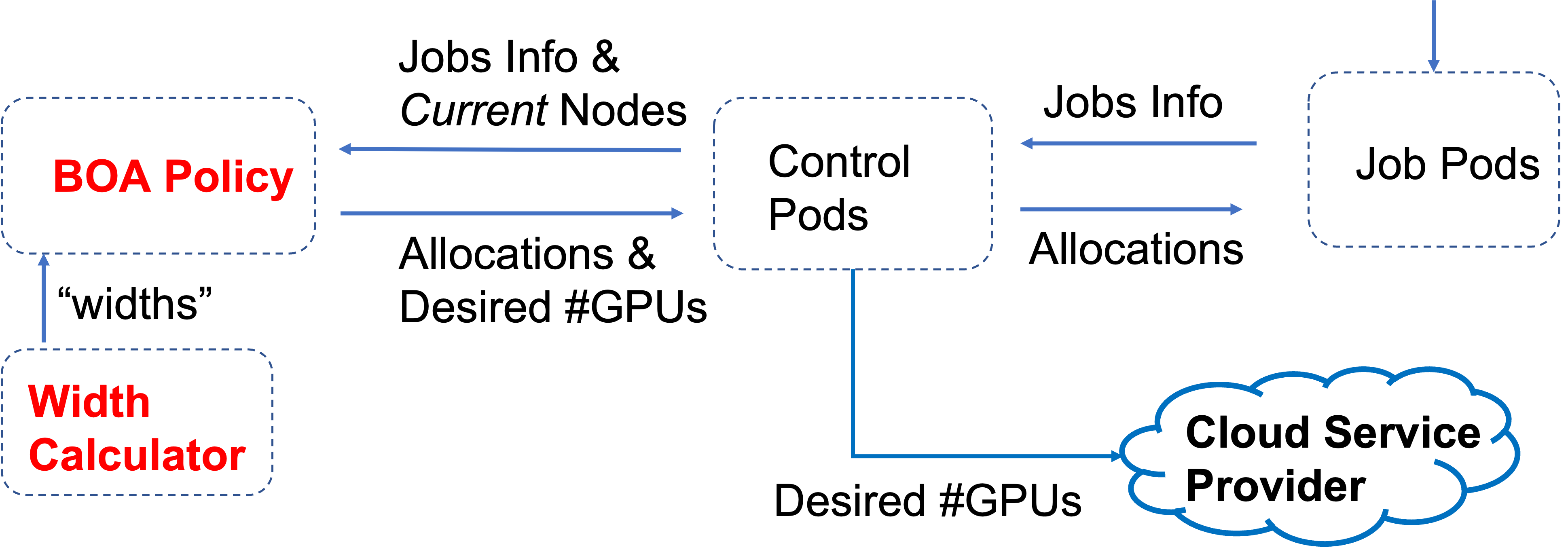}
    \vspace{-.1in}
    \caption{The \boa architecture. The \boa Width Calculator asynchronously refreshes $\{k^{(h)}_{ij}\}$, while the online scheduler executes the fixed-width policy.}
    \label{fig:fw}
    \vspace{-.2in}
\end{figure}

\subsection{Scheduler Overheads}
\label{sec:design:overheads}


\boa incurs low scheduling overhead.
Reading the width table takes $<1$~ms per refresh.
Solving the MICP~\eqref{eq:hetero-opt-rescale} fits easily in the $15$-minute refresh interval:
With 5 job types (see Section~\ref{sec:eval:setup}), the MICP for a homogeneous cluster runs in $4$ seconds and the MICP for a heterogeneous cluster with $3$ GPU types runs in $52$ seconds.\footnote{Measured on a single CPU core at $5\%$ optimality tolerance.}

Our implementation nonetheless uses the heuristic of Appendix~\ref{app:boa width}, which approximates the MICP by a sequence of convex optimizations, performs well in practice, and scales to workloads much larger than those we evaluate.

\paragraph{Rescaling overheads.}
{Rescaling, by contrast, is not cheap.
A rescaling checkpoints the job, terminates its pods, and restarts them on new nodes; \boa reuses AdaptDL's~\cite{adaptdl} rescaling mechanism, so the cost of a single rescaling is unchanged.
For a CIFAR-10 job we measured $\approx 20$ sec on a ``warm'' machine (container images cached) and $\approx 120$ sec on a ``cold'' one.
Of the $120$ sec, $\approx 75$ sec go to installing Python packages and initializing CUDA contexts and $25$ sec to loading the dataset from the Elastic File System (EFS); on a warm machine, half of the latency goes to connecting the distributed workers.} 

{As Section~\ref{sec:eval} shows, an occasional rescaling over a job's lifetime can greatly improve the cost-performance tradeoff, while frequent rescaling inflates job completion times; \boa rescales only when it pays.}

\section{Evaluation}
\label{sec:eval}

Our evaluation answers four questions.
(Q1) Does \boa improve the cost-performance Pareto frontier over state-of-the-art schedulers, and over autoscaling variants of these schedulers? (Section~\ref{sec:eval:sim});
(Q2) Do the simulation results hold in a real deployment? (Section~\ref{sec:eval:imple});
(Q3) Where does the improvement come from? (Section~\ref{sec:eval:gain});
(Q4) How robust is \boa to prediction error and workload variability? (Section~\ref{sec:eval:sensitivity}).
Section~\ref{sec:eval:setup} describes the three settings in which we evaluate \boa.
\begin{figure*}[ht]
    \centering
    \includegraphics[scale=0.189]{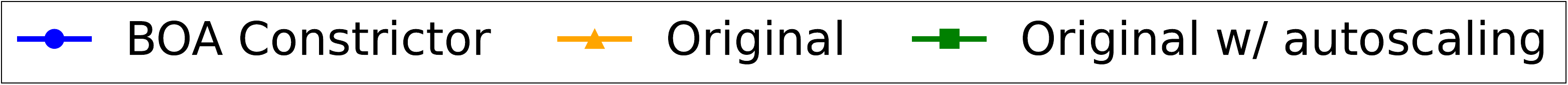}\\[2pt]
    \includegraphics[scale=0.189]{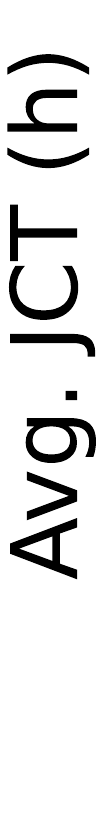}
    \begin{subfigure}[t]{0.28\linewidth}
        \centering
        \includegraphics[width=\textwidth]{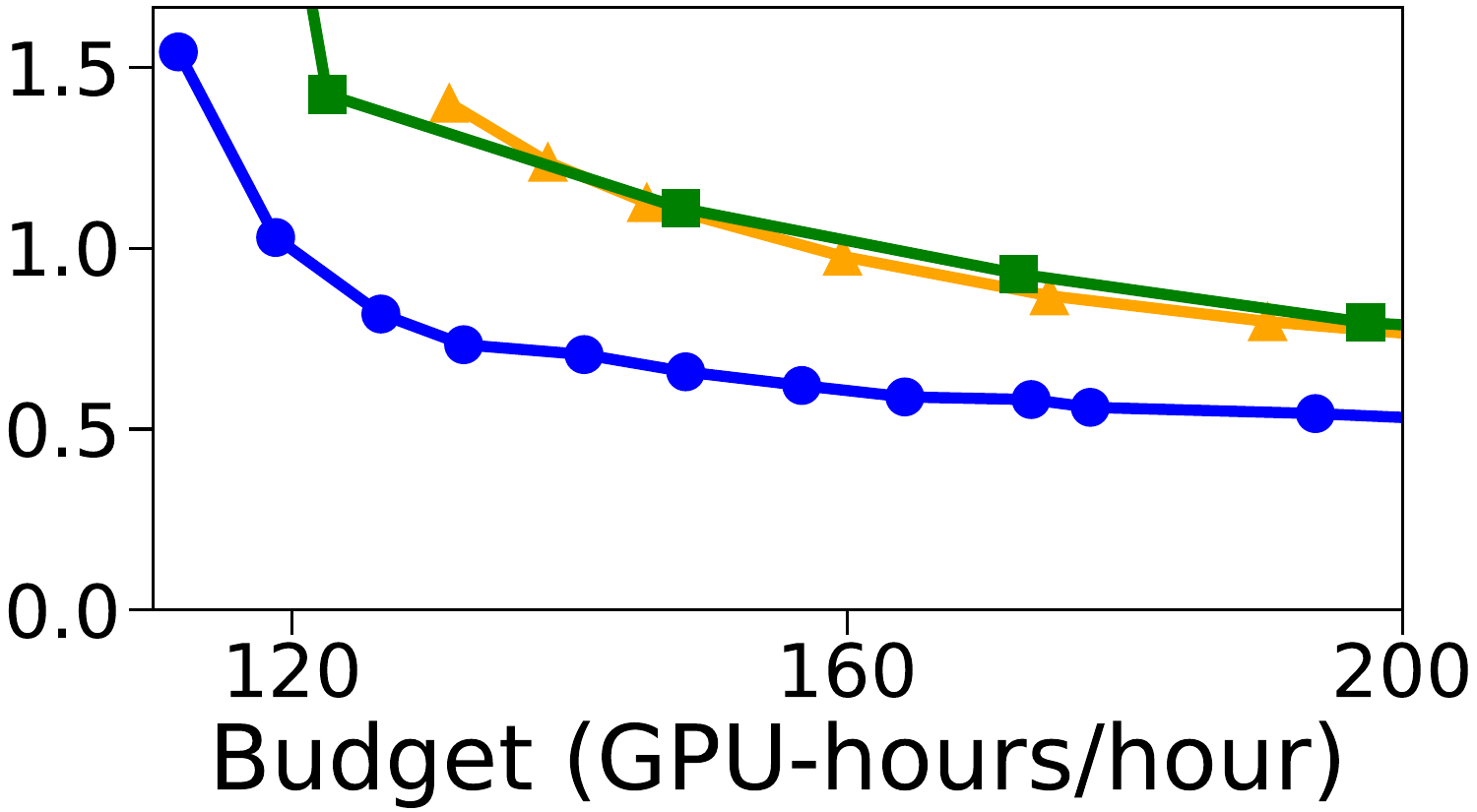}
        \caption{Pollux setting, average JCT.}
        \label{fig:sim:pollux:mean}
    \end{subfigure}
    \quad
    \begin{subfigure}[t]{0.28\linewidth}
        \centering
        \includegraphics[width=\textwidth]{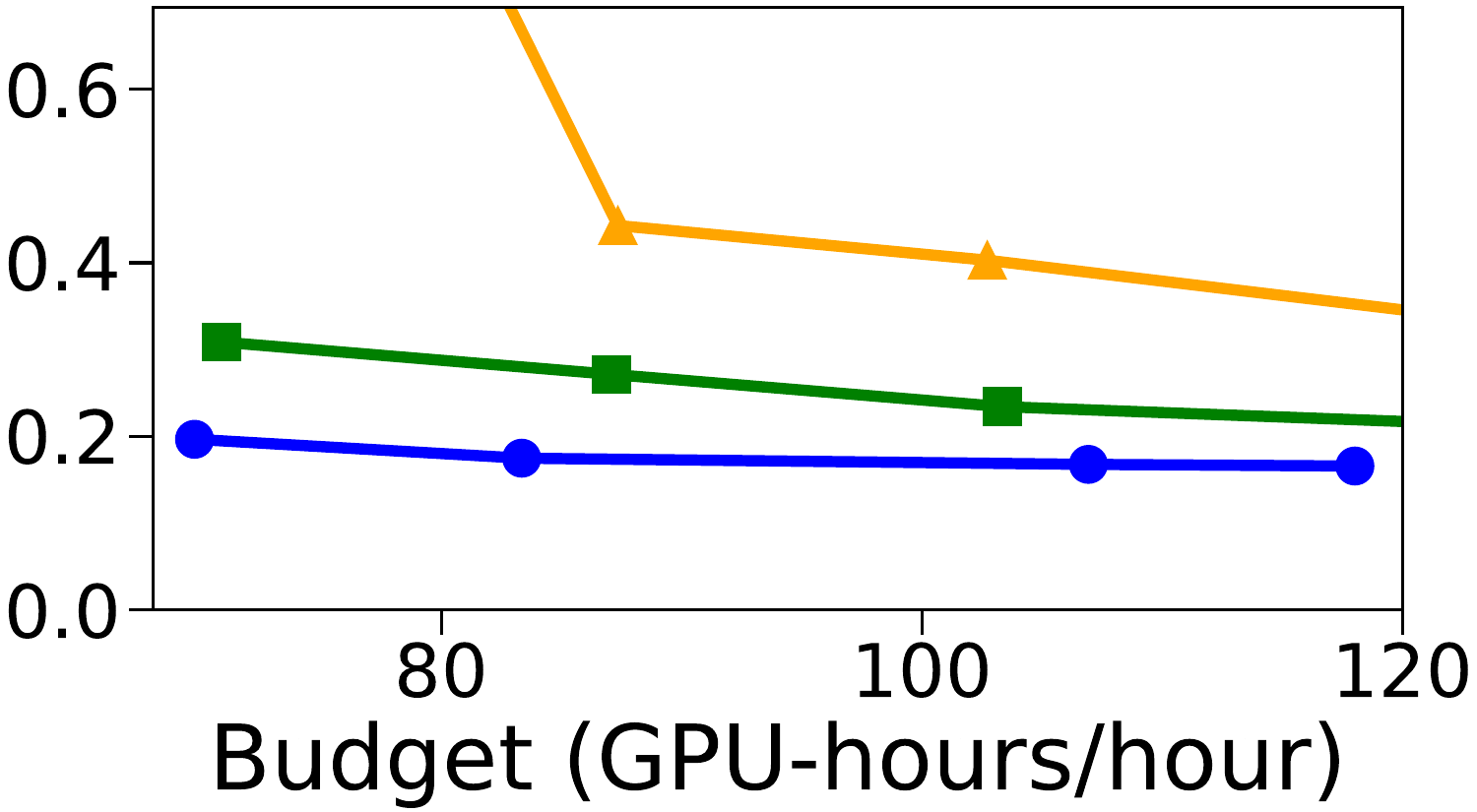}
        \caption{Rubick setting, average JCT.}
        \label{fig:sim:rubick:mean}
    \end{subfigure}
    \quad
    \begin{subfigure}[t]{0.28\linewidth}
        \centering
        \includegraphics[width=\textwidth]{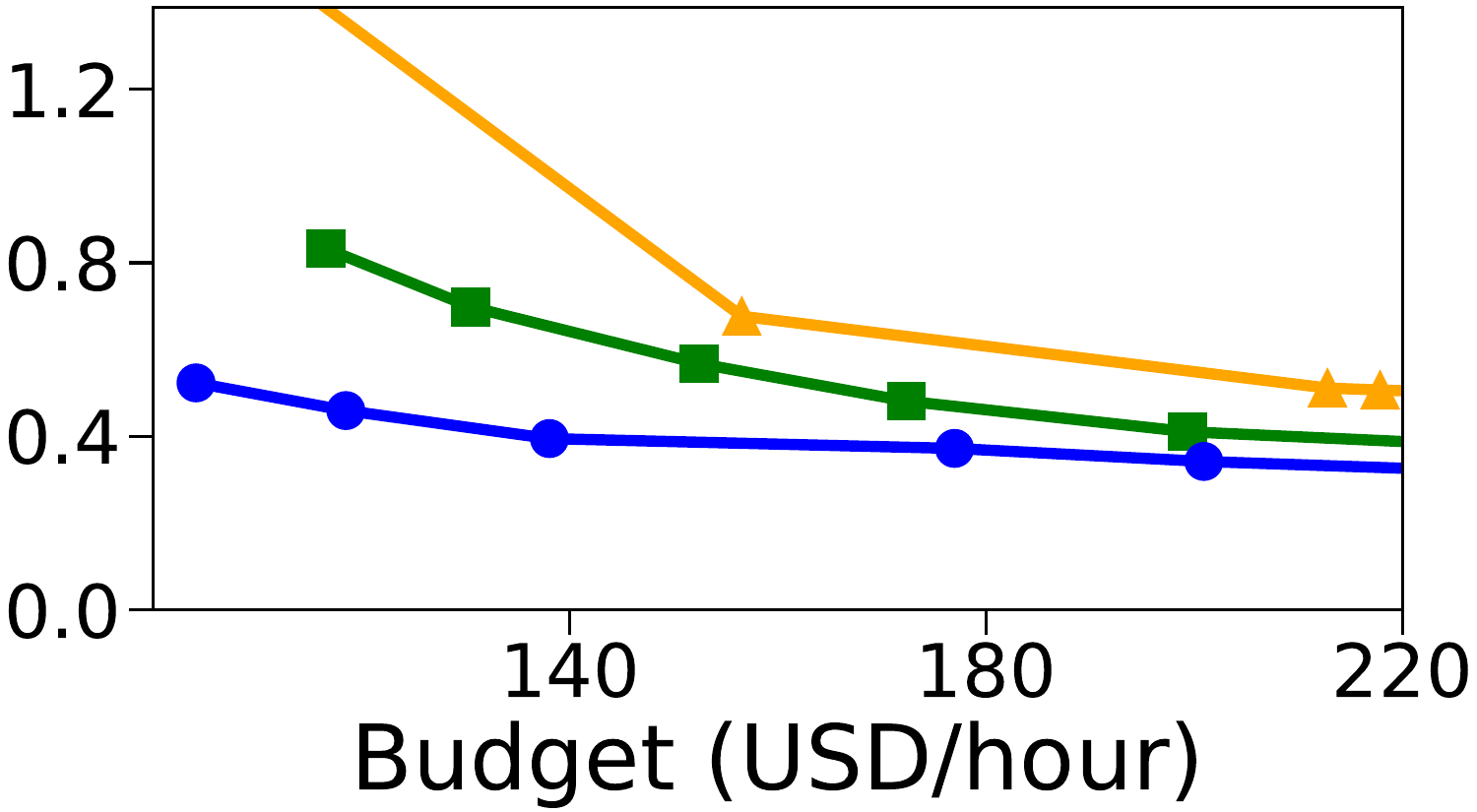}
        \caption{Sia setting, average JCT.}
        \label{fig:sim:heter:mean}
    \end{subfigure}
    \\
    \includegraphics[scale=0.189]{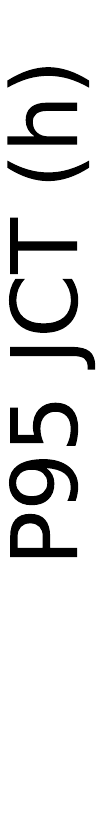}
    \begin{subfigure}[t]{0.28\linewidth}
        \centering
        \includegraphics[width=\textwidth]{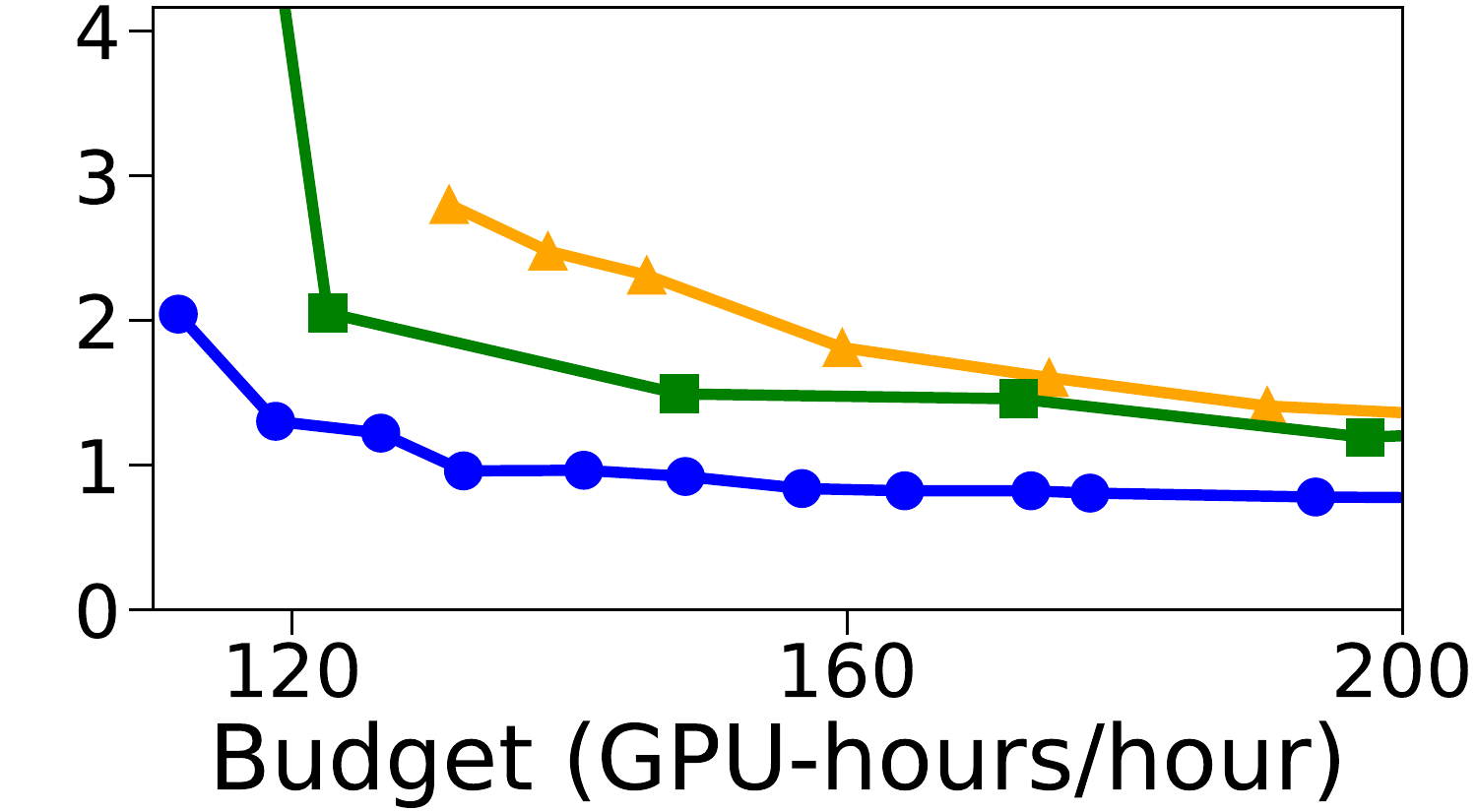}
        \caption{Pollux setting, P95 JCT.}
        \label{fig:sim:pollux:tail}
    \end{subfigure}
    \quad
    \begin{subfigure}[t]{0.28\linewidth}
        \centering
        \includegraphics[width=\textwidth]{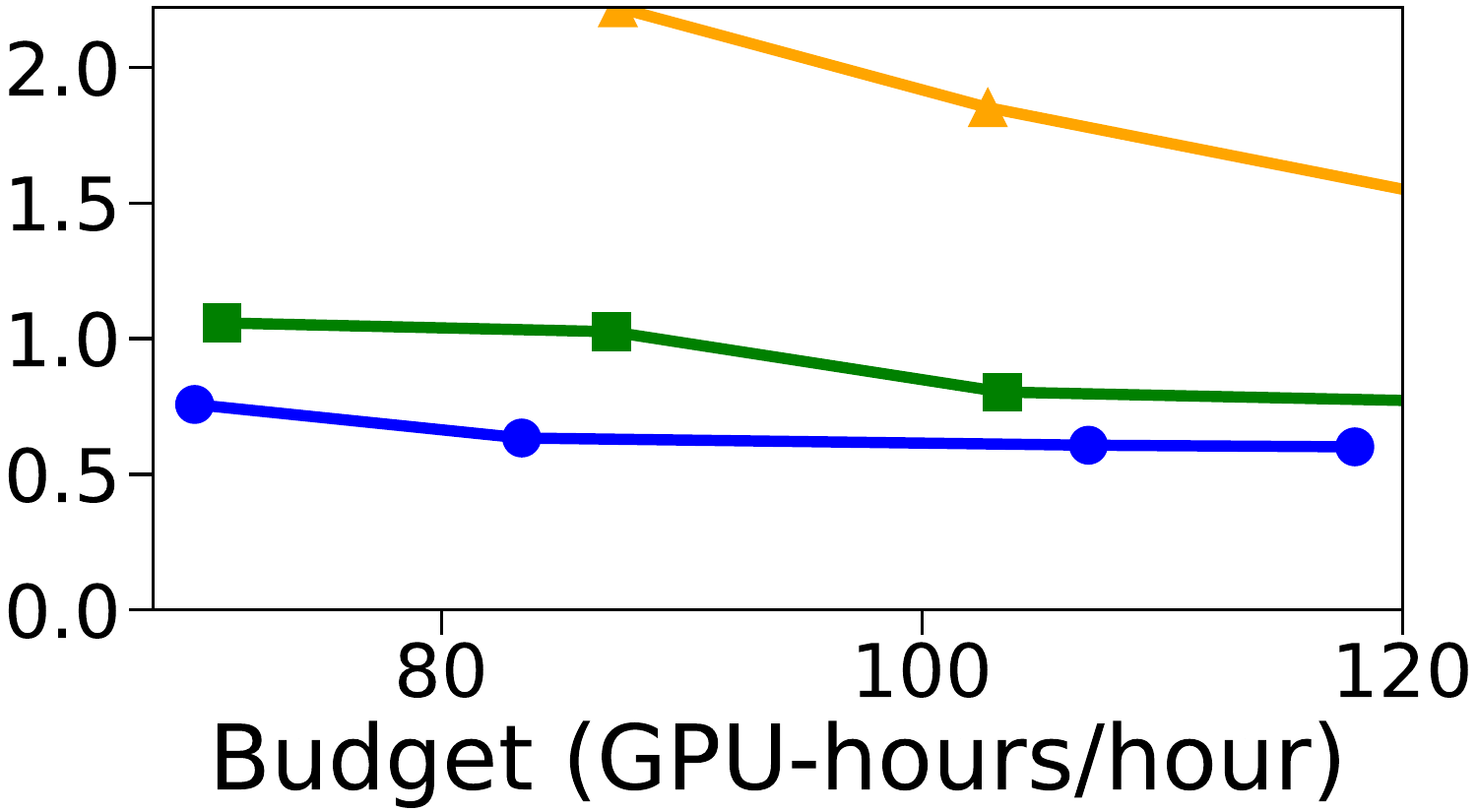}
        \caption{Rubick setting, P95 JCT.}
        \label{fig:sim:rubick:tail}
    \end{subfigure}
    \quad
    \begin{subfigure}[t]{0.28\linewidth}
        \centering
        \includegraphics[width=\textwidth]{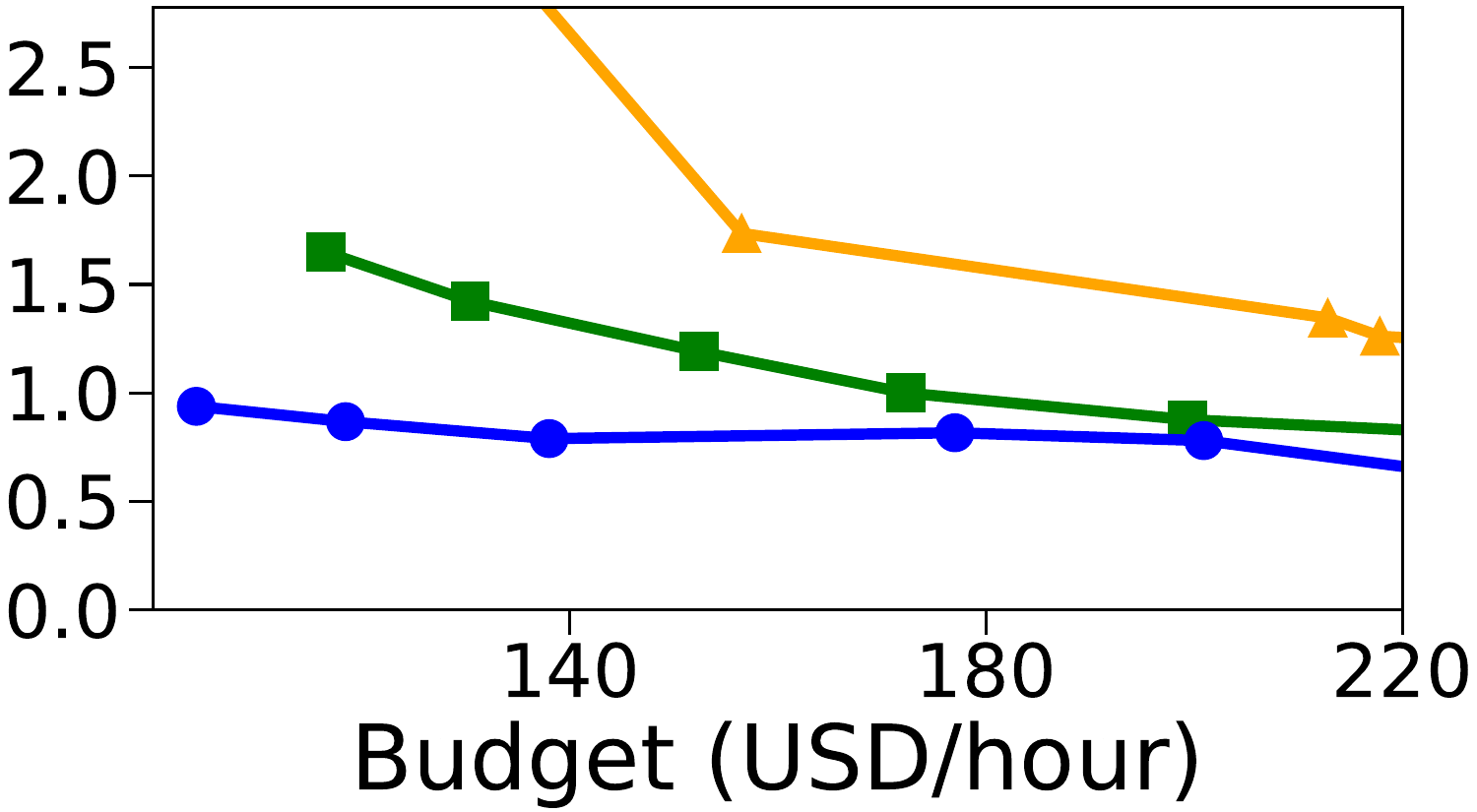}
        \caption{Sia setting, P95 JCT.}
        \label{fig:sim:heter:tail}
    \end{subfigure}
    \caption{\boa versus competitor policies in the Pollux setting (a, d), Rubick setting (b, e), and Sia setting (c, f). In each plot, \emph{Original} is the scheduler from the prior work (Pollux, Rubick, or Sia) on a fixed-size cluster. \emph{Original w/ autoscaling} is the autoscaling variant of each policy.
\boa improves the Pareto frontier of cost vs. average JCT (top) and P95 JCT (bottom) in every setting.}
\vspace{-.2in}
    \label{fig:sim}
\end{figure*}
\subsection{Experimental Setup}
\label{sec:eval:setup}

State-of-the-art cluster schedulers assume a fixed-size cluster with elastic training jobs, and continuously change job allocations and hyperparameters to adapt to the current state of the system.
We note that non-adaptive schedulers~\cite{narayanan2020heterogeneity, gu2019tiresias, peng2018optimus} are dominated by the adaptive systems, so we omit them from our evaluation.

The state-of-the-art adaptive schedulers largely adopt the same high-level approach.
All the schedulers allocate a fixed-size cluster of GPUs while jointly optimizing GPU allocations with some set of job hyperparameters.
These schedulers are mainly differentiated by the exact set of hyperparameters that the scheduler tunes.
Each scheduler then makes allocation decisions by using an optimizer to find the setting of hyperparameters and allocations that maximizes the \emph{cluster-wide efficiency} across jobs.
Here, cluster-wide efficiency measures the average speedup per GPU, see 
Appendix~\ref{app:autoscaling-extensions} for a full definition.
We evaluate three different state-of-the-art adaptive schedulers, which we differentiate based on the hyperparameters exposed to the scheduler.

\subhead{The Pollux~\cite{qiao2021pollux} setting} considers a homogeneous cluster of NVIDIA T4 GPUs where job batch size and learning rate are tuned adaptively.

\subhead{The Rubick~\cite{zhang2025rubick} setting} considers a homogeneous cluster of NVIDIA A800 GPUs where each job's \emph{parallel-training configuration} (e.g., data parallelism vs model parallelism) is chosen adaptively. 

\subhead{The Sia~\cite{subramanya2023sia} setting} considers a heterogeneous cluster of NVIDIA T4, V100, and A100 GPUs where batch-size is tuned adaptively.


Experiments for Pollux and Sia use the \textbf{newTrace}~\cite{subramanya2023sia}, comprising $960$ jobs sampled over $48$ hours from a production DL training system.
Experiments for Rubick use the \textbf{Rubick} workload, comprising $406$ jobs from the Microsoft Philly trace~\cite{jeon2019analysis}, including LLaMA-2-7B and LLaMA-30B jobs that cannot fit on a single GPU ($\base_i>1$).
As in the source papers, a job completes when it reaches its model's target accuracy (Pollux, Sia)~\cite{qiao2021pollux} or a fixed iteration count (Rubick)~\cite{zhang2025rubick}.
Appendix~\ref{app:setup} provides details on each setting's workload, hardware, and GPU pricing.
For each setting, we compare \boa against the proposed scheduler (i.e., Pollux, Rubick, or Sia). 


\subhead{Implementing the autoscaling variants.}
\label{sec:eval:autoscaling}
The above schedulers target a {\em fixed-size} cluster and are our primary baselines.
Pollux~\cite{qiao2021pollux} also proposed a na\"ive autoscaling variant of the Pollux policy that tunes the cluster size based on cluster-wide efficiency.
However, this variant was only evaluated on a single job, and its prototype in AdaptDL does not scale to realistic workloads (Appendix~\ref{app:autoscaling-extensions}).
For the sake of comparison we implement autoscaling variants of Pollux, Rubick, and Sia based on the proposal in the Pollux paper.

Specifically, because allocating many GPUs to a job is inefficient, cluster efficiency drops when there are few jobs in the system.
The Pollux paper proposes setting a target level of cluster efficiency and growing or shrinking the cluster to meet this target. 
We apply this autoscaling mechanism to the three competitor policies, tuning the target efficiency level to adjust the cost of serving the workload.

We give these autoscaling variants every benefit: finding the cluster size that meets the target efficiency requires a combinatorial search in every scheduling round, and we ignore this overhead in our simulations.  Note that \boa computes its cluster size in constant time.
We also tuned the variants' parameters and design choices (e.g., the scheduling quantum) to the best of our ability and evaluate the best-performing version (Appendix~\ref{app:autoscaling-extensions}).


\subhead{Inputs to the schedulers.}
Every policy we evaluate relies on a performance model that learns each class's speedup function and inherent work per epoch as jobs run, and that drives the policy's hyperparameter choices.
\boa reads $s_{ij}$ and $\E[X_{ij}]$ from the same model as its competitor, so no policy has an information advantage.


In practice, customers have long job histories, so our main experiments (Figures~\ref{fig:sim} and~\ref{fig:impl}) start every policy with a converged performance model and give \boa the trace's average arrival rates.
This isolates scheduling gains after both performance models and hyperparameter search have converged.
Section~\ref{sec:eval:sensitivity} shows that \boa performs equally well on reasonably long traces when starting without speedup profiles and estimating arrival rates during training.


\subhead{Operating budget.} For every curve in Figures~\ref{fig:sim} and~\ref{fig:impl}, the x-axis is the \emph{measured} time-average usage of the policy (GPU-hours per hour in the homogeneous settings, USD per hour in the heterogeneous one).
\boa takes the operating budget $b$ as an input, and its measured usage equals $b$.
The competitor policies have no notion of a budget: for the fixed-cluster schedulers we sweep the cluster size, and for the autoscaling variants we sweep the target cluster efficiency $c$, measuring the resulting usage after the fact.
Neither knob predicts the usage it will produce, so only \boa can be pointed directly at a desired budget.

\subsection{Simulation Experiments}
\label{sec:eval:sim}

We simulate \boa and its competitors in each setting using the simulator from the corresponding paper: the Sia simulator~\cite{subramanya2023sia} for the Pollux and Sia settings, and the Rubick simulator~\cite{zhang2025rubick} for the Rubick setting.
Both simulators model each job from profiling data, account for rescaling overheads, and capture inter-job interference.
{Figure~\ref{fig:sim} shows the cost-performance Pareto frontier in all three settings for both average JCT and P95 JCT.}


\boa improves the Pareto frontier in every panel of Figure~\ref{fig:sim}.
Against the published fixed-cluster schedulers, \boa achieves up to $1.9\times$, $2.6\times$, and $1.8\times$ lower average JCT at the same operating budget in the Pollux, Rubick, and Sia settings, respectively; equivalently, it needs up to $2\times$ less budget to reach the same average JCT. 

The autoscaling variants narrow the gap but do not close it: \boa still achieves up to $1.8\times$ lower average JCT than the best autoscaling variant at the same budget.

Although \boa optimizes average JCT, it also improves the tail.
Under each competitor, a single unlucky job may be rescaled many times, for example during a burst of arrivals.
Under \boa, a job's sequence of allocations is fixed by its class, so its number of rescalings is bounded by its number of epochs and does not depend on the other jobs in the system.
As a result, at a given budget, \boa reduces P95 JCT by up to $2.2$--$3.5\times$ relative to the fixed-cluster schedulers and by up to $1.9\times$ relative to their autoscaling variants.

\noindent\textbf{Takeaway.} \boa improves the cost-performance Pareto frontier in all three settings, for both average and tail JCT, against the published schedulers and against their autoscaling variants.

\subsection{Physical Cluster Experiments}
\label{sec:eval:imple}

{We deploy \boa on an AWS cluster in the Pollux setting to both evaluate our real-world implementation and to validate our simulations.}
We run $85$ jobs from the \textbf{workload-1} trace~\cite{qiao2021pollux} on \texttt{g4dn.12xlarge} instances ($4\times$~T4 per node; Appendix~\ref{app:setup}).

\begin{figure}[h]
    \centering
    \includegraphics[scale=0.16]{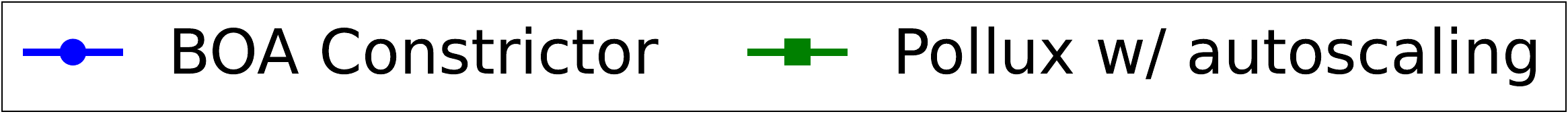}\\[2pt]
    \includegraphics[scale=0.17]{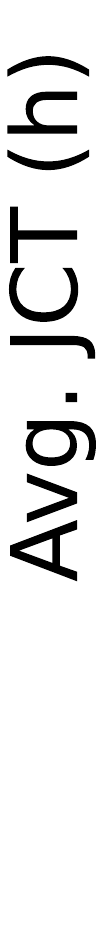}
    \begin{subfigure}[t]{0.42\linewidth}
        \centering
        \includegraphics[width=\linewidth]{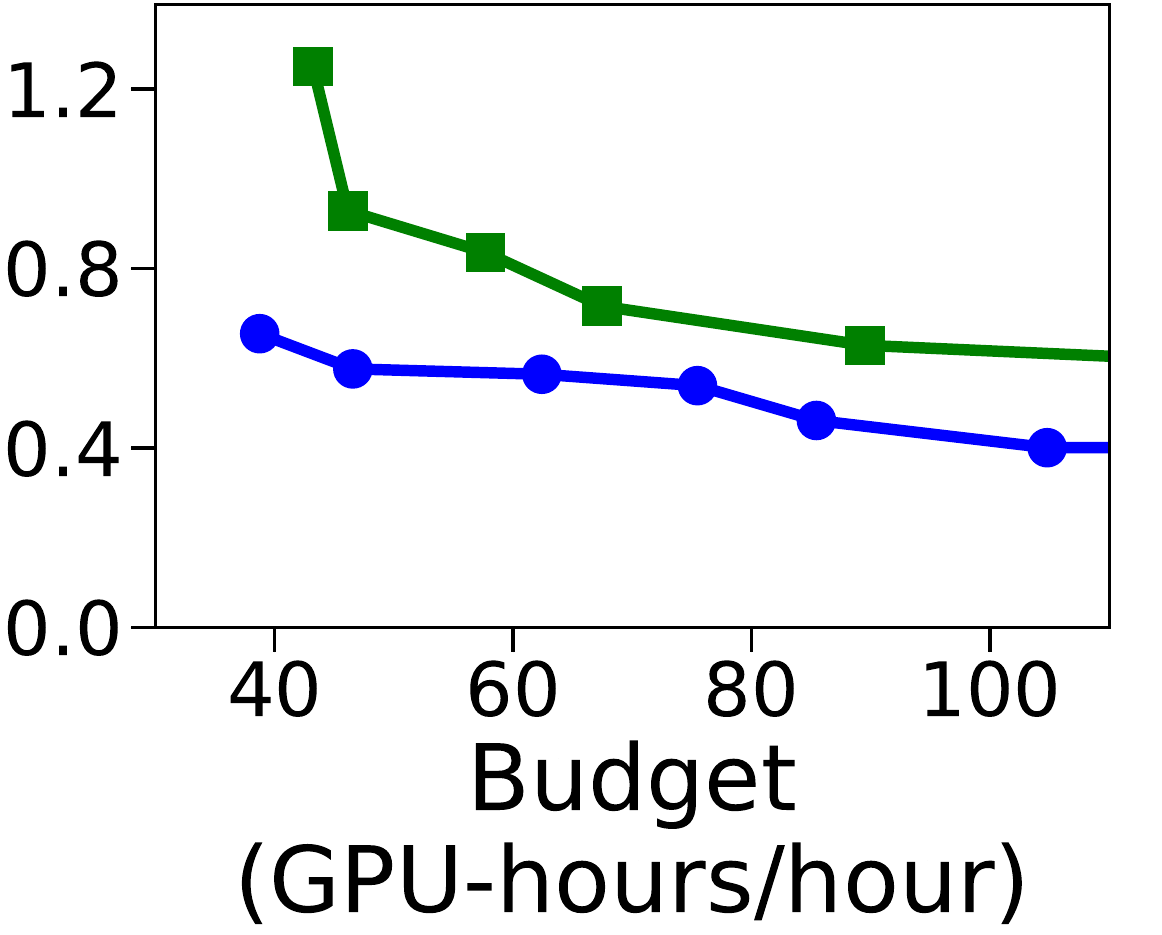}
        \caption{Implementation.}
        \label{fig:impl:impl}
    \end{subfigure}
    \hfill
    \begin{subfigure}[t]{0.42\linewidth}
        \centering
        \includegraphics[width=\linewidth]{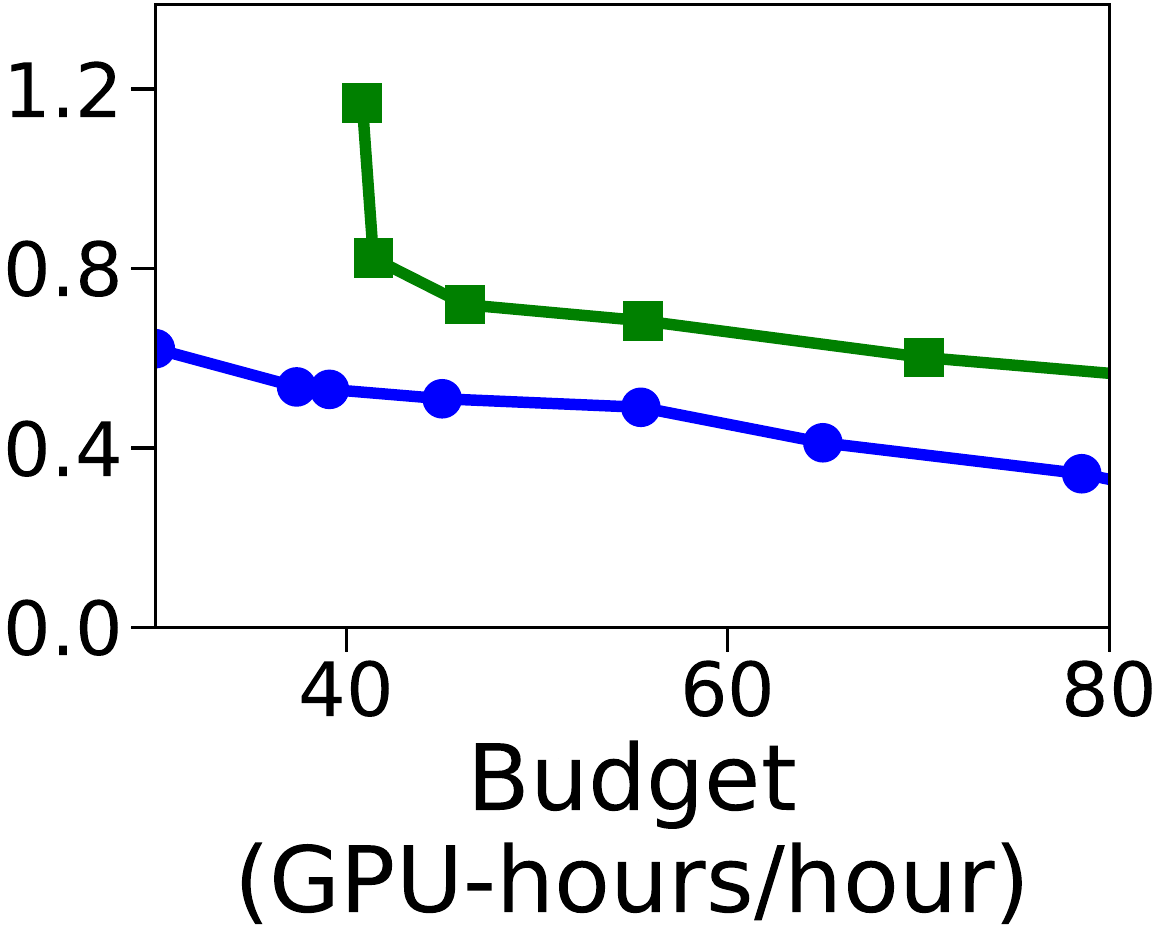}
        \caption{Simulation.}
        \label{fig:impl:sim}
    \end{subfigure}
    
    \caption{Comparing \boa and Pollux-with-Autoscaling running $85$ jobs from the \textbf{workload-1} trace~\cite{qiao2021pollux}. Simulation results match implementation results, showing that \boa significantly improves the Pareto frontier of average JCT vs.~operating budget.}
    \label{fig:impl}
\end{figure}
\subhead{Results.}
Figure~\ref{fig:impl} shows the implementation results alongside matching simulations.
In the implementation (Figure~\ref{fig:impl:impl}), \boa achieves up to $1.6\times$ lower average JCT, and it reaches an average JCT of ${\sim}0.65$ hours with ${\sim}2.1\times$ less budget.

The simulations (Figure~\ref{fig:impl:sim}) reproduce these results.
The small discrepancies have two sources:  {\em (i)}
The added variability of a real cluster makes both policies' speedup models less accurate, and {\em (ii)} a released node keeps being billed until the cloud provider reclaims it, a delay the simulator does not model.
Because these discrepancies affect both policies, the relative gap measured in implementation matches simulation, validating the simulations of Section~\ref{sec:eval:sim}.

At the lowest budget, Pollux-with-Autoscaling crashed: with so few GPUs its optimizer could not meet the target cluster efficiency, so it rescaled jobs continually as jobs accumulated in the system; we report its average JCT at the time of the crash.
The same failure appears in our low-budget simulations.
Efficiency-driven autoscaling thus breaks down at low budgets, where the efficiency target is unattainable, whereas \boa remains stable. 

\noindent\textbf{Takeaway.} \boa's gains hold on a real cluster. The simulator accurately predicts the relative gap between the policies.

\subsection{Where Does the Gain Come From?}
\label{sec:eval:gain}
We note three main advantages of \boa.

\subhead{Jobs rescale only when it pays.}
\boa computes the allocation for every job epoch offline and decides in advance which rescaling overheads are worth their cost (Section~\ref{sec:theory:rescaling}).
A job therefore changes allocations only at epoch boundaries for which rescaling is worthwhile.

By contrast, the competitor systems may rescale jobs frequently.
For example, Pollux-with-Autoscaling may re-optimize \emph{all} job allocations at every $60$-second scheduling quantum. 
Hence, small bursts of arrivals or completions can trigger widespread rescaling in the competitor systems, even when the cluster size remains fixed.
When using $\$135/\mbox{hr}$, Sia-with-Autoscaling rescales each job roughly $4$ times more frequently than \boa in simulation.
In AWS deployments with a budget of $65$ GPU-hours per hour, Pollux-with-Autoscaling rescales jobs about twice as frequently as \boa.
With rescaling latencies of $20$--$120$ seconds per event, this represents a significant advantage for \boa. 

\begin{figure}[ht]
    \centering
    \includegraphics[width=.85 \linewidth]{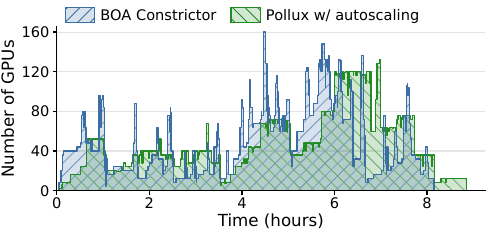}
    \vspace{-.1in}
    \caption{GPU usage of \boa and Pollux-with-Autoscaling implementations. \boa lowers JCTs by reacting faster and more aggressively to bursts in arrivals.}
    \label{fig:gpu usage}
\end{figure}

\subhead{More reactive cluster scaling.}
The fixed-width structure of the BOA policy has additional benefits when it comes to autoscaling.
The BOA policy says to scale the cluster up or down only in response to the fixed-width demands of specific jobs.
This means that when the cluster is scaled up, BOA identifies a small set of jobs to rescale and run on the newly rented nodes.
All other jobs in the system run without rescaling.
Scaling the cluster size is highly disruptive for the efficiency-based autoscaling variants.
These na\"ive policies can detect a potential benefit to having a larger cluster, but they do not know \emph{which specific allocations} to increase.
As a result, they re-optimize \emph{all} job allocations, incurring a potentially large number of job rescalings.
An analogous process plays out when scaling down the cluster size.

This difference is clear in Figure~\ref{fig:gpu usage}, which shows how cluster size varies over time under \boa and Pollux-with-Autoscaling, each using 65 GPU-hours per hour.
Pollux-with-Autoscaling scales the cluster size conservatively to avoid excessive rescaling overhead, and is less reactive to transient changes in the workload.
Because the fixed-width nature of the BOA policy largely decouples rescaling overheads from cluster scaling decisions, \boa is more responsive to workload changes, resulting in higher peak cluster size and more fine-grained scale-up and scale-down decisions.


\subhead{Cluster efficiency is the wrong objective.}
Our evaluation validates our theoretical results: cluster-wide efficiency is simply the wrong metric for balancing the cost-performance tradeoff.
Our theoretical optimality results in Section \ref{sec:BOA} show that optimizing cluster-wide efficiency is suboptimal compared to directly optimizing average JCT.
Similarly, optimizing average JCT subject to a budget constraint is \emph{not equivalent} to optimizing average JCT subject to a given target efficiency level.
To make this clear, we ran additional simulations that set job rescaling times to 0.
Even in this idealized case, where job rescaling and cluster scaling overheads were not an issue, \boa reduced the training budget required to achieve a given average JCT by up to $1.4\times$ compared to Pollux-with-Autoscaling.

\noindent\textbf{Takeaway.} Optimizing cluster efficiency is less effective than directly balancing the cost-performance tradeoff. \boa balances this tradeoff while limiting individual job rescalings \emph{and} scaling overall cluster size aggressively.

\subsection{Sensitivity Analysis}
\label{sec:eval:sensitivity}

\subhead{Starting without profiles or arrival rates.}
Our main experiments start every policy with a converged performance model and arrival rate estimates (Section~\ref{sec:eval:setup}).
Appendix~\ref{app:subtrace} repeats the Pollux-setting simulation without either.
{Every policy, including \boa, starts with an empty performance model and learns speedup functions and inherent work from scratch as jobs run.
\boa also refines its arrival rate estimate online as jobs arrive (Appendix~\ref{app:subtrace}).}

Under this setting, the gains persist (Figure~\ref{fig:coldstart}): at a given budget, \boa reduces average JCT by up to $2.2\times$ relative to Pollux and by up to $1.7\times$ relative to Pollux-with-Autoscaling, in line with Figure~\ref{fig:sim}(a, d).

\subhead{Speedup-prediction error.}
Appendix~\ref{app:subtrace} evaluates each policy's sensitivity to speedup-prediction errors.
\boa's average JCT is essentially unchanged between estimated and ground-truth speedup functions, whereas Pollux-with-Autoscaling's average JCT increases by up to $1.4\times$.

\boa is insensitive to these prediction errors because of its fixed-width structure.
A speedup prediction error shifts \boa's computed policy slightly, but does not result in significantly more rescaling.
By contrast, a reactive policy like Pollux rescales frequently when prediction errors cause the system to miss its target efficiency.
 
\begin{figure}[h]
    \centering
    \includegraphics[width=0.7 \linewidth]{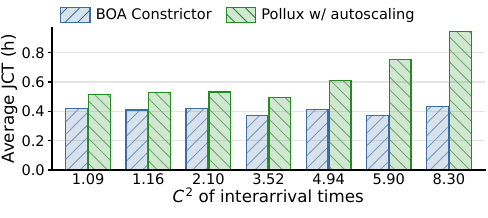}
    \vspace{-.2in}
    \caption{Comparison of the sensitivity of \boa and Pollux-with-Autoscaling to arrival process variability. Here, both policies use a budget of ${\sim}120$ GPUs.}
    \label{fig:time:poisson}
    \vspace{-.15in}
\end{figure}

\subhead{Workload variability.}
To test burstier workloads, we generate synthetic traces from \textbf{newTrace} that maintain the same mix of jobs and long-run arrival rate, but concentrate arrivals into ever shorter and more intense bursts.
We quantify burstiness using the squared coefficient of variation of interarrival times, $C^2$.
Figure~\ref{fig:time:poisson} shows that \boa's advantage over Pollux-with-Autoscaling grows with $C^2$.
Sharper bursts amplify the cost of Pollux-with-Autoscaling's slow, conservative autoscaling compared to \boa.

\section{Conclusion}
Training ML models in the cloud gives rise to a cost-performance tradeoff: additional GPUs reduce job completion times, but they do so at a price.
We derive the \policybest policy, which explicitly optimizes the average JCT across a stream of training jobs for any given budget constraint on average GPU usage.
Using this policy, \boa optimizes the cost-performance tradeoff in homogeneous clusters, heterogeneous clusters, and for a wide range of distributed training workloads.

We see two natural extensions of the \policybest policy. First, the \policybest policy targets a long-run average budget, which is natural for many customers.
However, some customers may require a hard budget constraint over a shorter time horizon (e.g., spend an exact amount in one week).
Second, while \policybest considers GPUs rented on demand, there are also reserved GPUs that can be rented long-term at a discount and spot GPUs that are interruptible and subject to price fluctuations.
We establish a practical theoretical framework that might be adapted to tackle these complex scenarios.
\clearpage


\section*{Acknowledgments}
This work was supported by the National Science Foundation under grants
CIF-2403194, III-2322973, CMMI-2307008, CIF-2403195, and III-2322974,
and by a Google Junior Faculty Award.

\bibliographystyle{mlsys2025}
\bibliography{bib,bibshort}

\section*{AI Usage Disclosure}

All of the essential intellectual work in this paper was done by the authors without AI tools.
This covers the idea and problem formulation, the model and all theoretical results and proofs, and the design of \boa.

We used generative AI tools in only two places: writing code and editing language.

\paragraph{Code.}
We used an AI coding assistant (Claude Code) to help write code, including parts of the implementation, the simulation and experiment scripts, and the scripts that produce the figures.
The authors specified what each piece of code should do.
Every line of AI-generated code was then checked line by line by the authors before it was used to produce any result in this paper.

\paragraph{Writing.}
In writing, AI tools were used to fix grammar.

\clearpage

\appendix
\section{Proof in offline setting}
\label{app:offline}

Throughout this appendix, we consider the setting of homogeneous GPUs without rescaling overheads. Our main goal is to prove the following theorem, which derives the BOA policy in this setting.
\begin{theorem}
We define \policybest-no-rescaling as the fixed-width policy whose widths $\{k_{ij}\}$ are given by the solution of
\begin{equation}
\begin{aligned}
& \underset{\{k_{ij}\}}{\min}
& & \frac{1}{\lambda}\sum_{i,j}\frac{\rho_{ij}}{s_{ij}(k_{ij})} \\
& \text{subject to}
& & \sum_{i,j}\frac{\rho_{ij}\,k_{ij}}{s_{ij}(k_{ij})}\leq b, \\
& & & k_{ij}\geq \base_i.
\end{aligned}
\label{eq:online-opt-norescale}
\end{equation}

\policybest-no-rescaling is the budget-optimal allocation policy for homogeneous clusters with no rescaling overhead, and this policy can be computed by solving a convex optimization problem.
\label{thm:hmo:summary}
\end{theorem}
This appendix proves Theorem \ref{thm:hmo:summary} via a series of lemmas. 
Our argument uses an offline auxiliary problem as a technical tool: in the offline problem, the arrival times and inherent work of every job at every epoch are known to the customer a priori. We characterize the optimal offline policy on every \emph{well-behaved} sample path (defined below) and then lift this offline result to the online setting.

\subsection{Offline auxiliary problem and well-behaved sample paths}

A sample path $\arrival$ is an infinite sequence of arrival times and inherent works, where $x^{(\ell)}_{ij}$ denotes the inherent work of the $\ell$th arriving type-$i$ job in epoch $j$. Let $n_i(t)$ denote the number of type-$i$ arrivals by time $t$, and $n(t):=\sum_{i=1}^M n_i(t)$ the total number of arrivals.

\begin{definition}[Well-behaved sample path]
\label{def:well-behaved}
A sample path $\arrival$ is \emph{well-behaved} if its time-average arrival rates and time-average inherent work converge to their finite means:
\begin{gather*}
\lambda_i \;=\; \lim_{t\to\infty}\frac{n_i(t)}{t},
\\
\E[X_{ij}] \;=\; \lim_{t\to\infty}\frac{\sum_{\ell=1}^{n_i(t)}x_{ij}^{(\ell)}}{n_i(t)} \quad\forall i,j.
\end{gather*}
\end{definition}

The finite-means assumption stated in Section~\ref{sec:model:problem} implies that a random sample path is well-behaved with probability $1$, so it suffices to derive a policy that is optimal for the set of well-behaved sample paths.

\subsection{Structural lemmas}

\begin{lemma}[No queueing]
\label{lemma:no queue}
For any well-behaved sample path $\arrival$, no job queues under the optimal policy. 
\end{lemma}
\begin{proof}
Assume there is a job queueing in the optimal policy. 
Suppose the job is allocated $k(t)$ GPUs at time $t$. 
Then, we create a policy $\pi'$ that removes all the queueing from the optimal policy. Mathematically, at any time $t$, $\pi'$ allocates the job $k(t')$ GPUs where $t'$ is the smallest time that satisfies $t=t'-q(t')$, where $q(t')$ is the job's total queueing time under the optimal policy before time $t'$.

\begin{figure}[h]
    \centering
        \centering
        \includegraphics[width=0.6\linewidth]{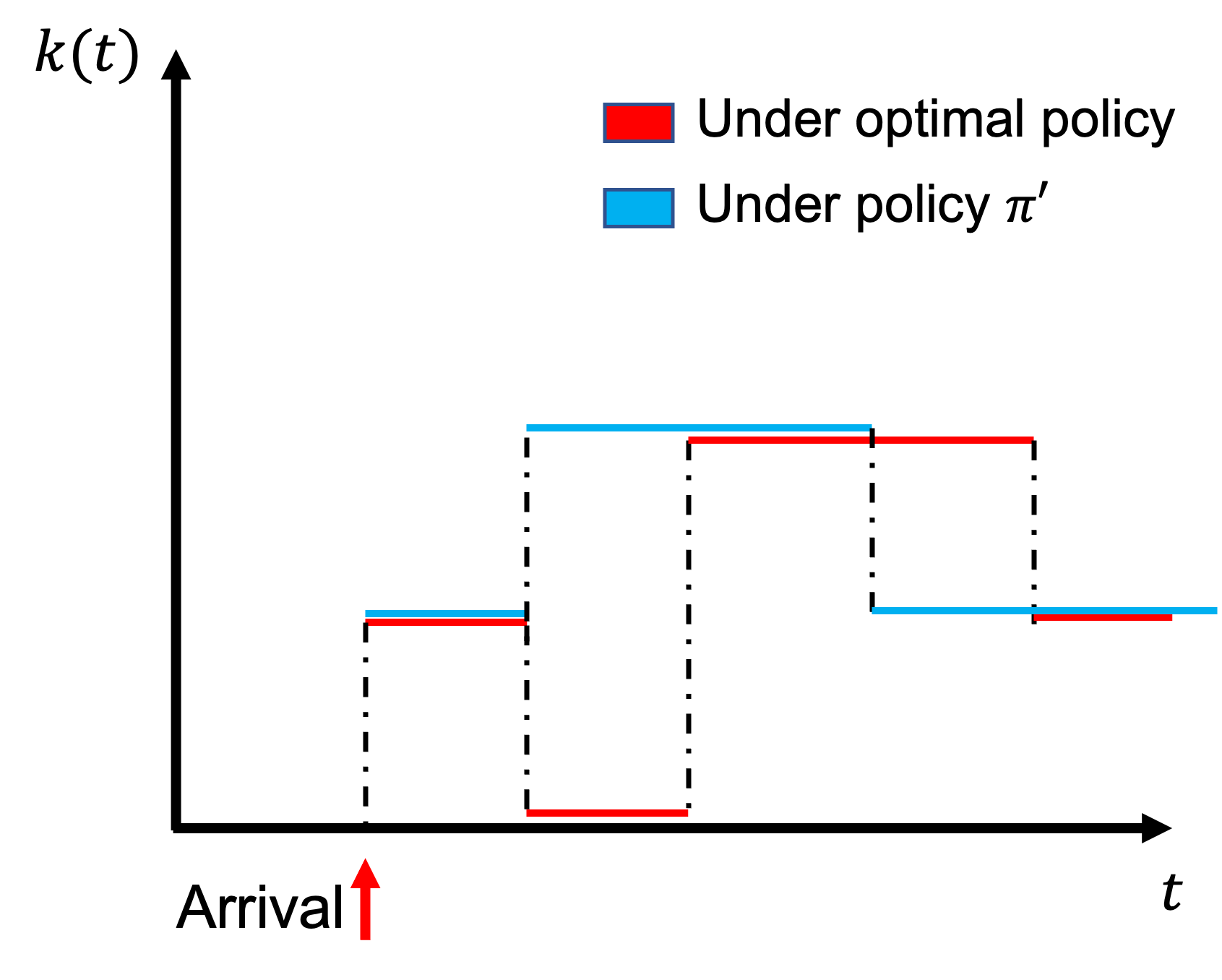}
    \caption{Illustration for the construction of policy $\pi'$.}
    \label{fig:noq}
\end{figure}

Removing its waiting time makes the mean response time lower but leaves the total GPU-hours the same, a contradiction. Note that for the online problem, we cannot create such a policy $\pi'$ because we do not know what happens in the future.
\end{proof}

\begin{lemma}[Same allocation for same speedup]
\label{lemma:fw}
Under any well-behaved sample path $\arrival$, any type-$i$ job at epoch $j$ is always assigned the same number of GPUs, and this allocation does not change until the epoch completes.
\end{lemma}
\begin{proof}
    Suppose the optimal policy $P$ does not satisfy the statement in the lemma. 
    Then either $P$ assigns $k_1\neq k_2$ GPUs to different type-$i$ jobs at epoch $j$, or $P$ assigns $k_1\neq k_2$ GPUs to a single type-$i$ job at epoch $j$ at different times. Let $x_1,x_2$ be the work completed using $k_1$ and $k_2$ GPUs, respectively.

    We construct a policy $P'$ which uses the same GPU-hours, but has lower mean response time than $P$.
    $P'$ is identical to $P$ except for the work $x_1$ and $x_2$:
    Instead of using two different assignments $k_1$ and $k_2$, $P'$ will choose a constant number of GPUs, $k$, to use in both instances.
    Let $t_1=\frac{x_1}{s_{ij}(k_1)}$ and $t_2=\frac{x_2}{s_{ij}(k_2)}$ be the durations of each of the GPU assignments under $P$.
    We choose $k$ to be the time average of the two assignments by setting
    \begin{equation}
        k= k_1\cdot \frac{t_1}{t_1+t_2} + k_2 \cdot \frac{t_2}{t_1+t_2}.
        \label{eq:fix-width proof eq1}
    \end{equation}

    The concavity of $s_{ij}(k)$ implies that the total time to complete $x_1$ and $x_2$ is lower under $P'$.  
    To see this, consider the average rate of work completion when processing $x_1$ and $x_2$ under both policies.
    The average work rate under $P$ is 
    $$\frac{t_1}{t_1+t_2}\cdot s_{ij}(k_1) + \frac{t_2}{t_1+t_2}\cdot s_{ij}(k_2).$$
    The average work rate under $P'$ is $s_{ij}(k)$.
    By concavity, we have
    \begin{align*}
        &\frac{t_1}{t_1+t_2}\cdot s_{ij}(k_1) + \frac{t_2}{t_1+t_2}\cdot s_{ij}(k_2)\\&\leq s_{ij}\left(\frac{t_1}{t_1+t_2}\cdot k_1 + \frac{t_2}{t_1+t_2}\cdot k_2\right)=s_{ij}(k).
    \end{align*}
    The total time to process $x_1$ and $x_2$ can be computed as the total work ($x_1 +x_2$) divided by the average work rate.
    We thus have that $P'$ completes the $x_1 + x_2$ work sooner than $P$.
    As a result, $P'$ has a lower mean response time than $P$.

    It is easy to see that $P'$ does not use more total GPU-hours than $P$.
    Specifically, note that $P$ uses $k_1 t_1 + k_2 t_2$ GPU-hours to process $x_1$ and $x_2$.
    Let $t'_1$ and $t'_2$ be the times required to process $x_1$ and $x_2$ respectively under $P'$.  Then the GPU-hours used to process $x_1$ and $x_2$ under $P'$ is
    $$(t'_1 +t'_2)\left(\frac{k_1t_1}{t_1+t_2} + \frac{k_2t_2}{t_1+t_2}\right).$$
    We have already shown that $t'_1 + t'_2 \leq t_1+t_2$, giving 
    \begin{align*}
        &(t'_1 +t'_2)\left(\frac{k_1t_1}{t_1+t_2} + \frac{k_2t_2}{t_1+t_2}\right)\\
        &\leq(t_1 +t_2)\left(\frac{k_1t_1}{t_1+t_2} + \frac{k_2t_2}{t_1+t_2}\right)= k_1t_1+k_2t_2.
    \end{align*}

    This leads to the contradiction, in that the policy $P'$ is using less budget but achieving better mean response time. 
\end{proof}

Lemmas~\ref{lemma:no queue} and \ref{lemma:fw} together imply that the optimal offline policy is a \emph{fixed-width} policy in the sense of Definition~\ref{def:fixedwidth}: every job is allocated some number of GPUs immediately upon arrival, and its allocation depends only on its class and current epoch.

\subsection{Operating budget of fixed-width policies}

The following lemma develops an alternate formulation of the operating budget of a fixed-width policy and matches the closed-form expressions used in Section~\ref{sec:theory:hmo}.

\begin{lemma}[Operating budget of a fixed-width policy]
\label{lemma:budget of fixed width}
    Given a well-behaved sample path $\arrival$, the operating budget of a fixed width policy with parameters $\{k_{ij}\}$ is
    \begin{align*}
    \bar{B} :=\lim_{t\to\infty} \frac{\int_0^tK(s)ds}{t} \ &\stackrel{(a)}{=} \ \lim_{t\to\infty} \frac{\sum_{i=1}^{n(t)}B^{(i)}}{t} \\
    &\stackrel{(b)}{=}  \ \sum_{i,j} \frac{\rho_{ij}k_{ij}}{s_{ij}(k_{ij})}.
    \end{align*}
    Here $B^{(i)}$ is defined to be the GPU-hours used to complete the $i^{th}$ arriving job.
\end{lemma}

\begin{proof}

    Part (a) of our claim says that tracking the total GPU usage at every time $t$ is equivalent to tracking the GPU-hours used to process each job, $B^{(i)}$. 

    To prove part (b), we show that we can take the limit of this equivalent formulation to prove our claim.
    Note that the fixed width policy assigns $k_{ij}$ GPUs to any type-$i$ job in epoch $j$.
    Hence, the GPU-hours spent on the $\ell^{th}$ type-$i$ job is $\sum_j \frac{x^{(\ell)}_{ij}k_{ij}}{s_{ij}(k_{ij})}$.

    Thus we have that 
    \begin{align*}
        &\lim_{t\to\infty} \frac{\sum_{i=1}^{n(t)}B^{(i)}}{t} \\
        &\quad=\lim_{t\to\infty} \frac{\sum_{i=1}^{M} \sum_{\ell=1}^{n_i(t)} \sum_j \frac{x^{(\ell)}_{ij}k_{ij}}{s_{ij}(k_{ij})}}{t} \\
        &\quad=\sum_{i=1}^M \sum_j \frac{k_{ij}}{s_{ij}(k_{ij})} \left(\lim_{t\to\infty}
        \frac{\sum_{\ell=1}^{n_i(t)}x^{(\ell)}_{ij}}{t}
        \right),
    \end{align*}
    where
    \begin{align*}
        \lim_{t\to\infty}
        \frac{\sum_{\ell=1}^{n_i(t)}x^{(\ell)}_{ij}}{t} &= \lim_{t\to\infty}
        \frac{\sum_{\ell=1}^{n_i(t)}x^{(\ell)}_{ij}}{n_i(t)} \frac{n_i(t)}{t} \\
        &= \lambda_i\E[X_{ij}] = \rho_{ij}.
    \end{align*}
    
\end{proof}

\subsection{Optimality of \policybest-no-rescaling and lift to the online setting}

We now show that \policybest-no-rescaling is offline optimal for any well-behaved sample path $\arrival$.

\begin{lemma}
\label{lemma:policybest-offline}
    For any well-behaved sample path $\arrival$, \policybest-no-rescaling is the optimal offline policy.
\end{lemma}
\begin{proof}
    Lemmas \ref{lemma:no queue} and \ref{lemma:fw} show that the optimal offline policy is a fixed-width policy. Thus, it suffices to show that \policybest-no-rescaling is the optimal offline fixed-width policy.

    For any job in $\arrival$ of type $i$ and size $x_{ij}$ at epoch $j$, the JCT under a fixed width policy parameterized with $\{k_{ij}\}$ is $\sum_j \frac{x_{ij}}{s_{ij}(k_{ij})}$. 
    Thus, we have that

    \begin{align*}
        \E[T]&:=\lim_{t\to\infty} \frac{\sum_{i=1}^{n(t)}T_i}{n(t)}\\&=\lim_{t\to\infty}\sum_{i=1}^M \frac{\sum_{\ell=1}^{n_i(t)} \sum_j x^{(\ell)}_{ij}}{s_{ij}(k_{ij})n(t)}\\&=\sum_{i=1}^M \sum_j \frac{1}{s_{ij}(k_{ij})} \left(
    \lim_{t\to\infty} \frac{\sum_{\ell=1}^{n_i(t)}x^{(\ell)}_{ij}}{n(t)}
    \right)\\
    &=\frac{1}{\lambda}\sum_{i=1}^M \frac{\sum_j \rho_{ij}}{s_{ij}(k_{ij})}.
    \end{align*}
    Moreover, by Lemma~\ref{lemma:budget of fixed width}, the operating budget is $\sum_{i,j} \frac{\rho_{ij}k_{ij}}{s_{ij}(k_{ij})}$. Thus solving the optimal set of $k_{ij}$ is equivalent to solving the following optimization problem:
    \begin{equation*}
        \begin{aligned}
        & \underset{k_{ij}} {\min}
        & & \frac{1}{\lambda}\sum_{i,j}\frac{  \rho_{ij} }{s_{ij}(k_{ij})}\\
        & \text{subject to}
        & & \sum_{i,j}\frac{\rho_{ij} k_{ij}}{s_{ij}(k_{ij})}\leq b.\\
        \end{aligned}
    \end{equation*}

This is the optimization problem~\eqref{eq:online-opt-norescale} that defines \policybest-no-rescaling. This shows that \policybest-no-rescaling is the optimal fixed-width policy. 
\end{proof}

\medskip
\noindent\textbf{Lift to the online setting.} We now complete the proof of Theorem~\ref{thm:fw-optimal} by lifting Lemma~\ref{lemma:policybest-offline} to the online problem. The optimization~\eqref{eq:online-opt-norescale} that defines \policybest-no-rescaling depends only on the load parameters $\{\rho_{ij}\}$, which are known to the customer in the online setting. Hence \policybest-no-rescaling is implementable online. By Lemma~\ref{lemma:policybest-offline}, \policybest-no-rescaling is optimal on every well-behaved sample path; since the finite-means assumption implies that the set of sample paths that are not well-behaved has measure zero, \policybest-no-rescaling is optimal in the online setting almost surely.

\subsection{Convex optimization Translation}
\label{sec:appendix convex}

In this appendix, we rewrite the optimization problem \eqref{eq:online-opt-norescale} into a convex optimization problem. The problem~\eqref{eq:online-opt-norescale} is not immediately convex because the constraint terms $\frac{k_{ij}}{s_{ij}(k_{ij})}$ are not necessarily convex. However, by applying a change of variables, we can translate it into a convex optimization problem.

\begin{theorem}
    The optimization problem \eqref{eq:online-opt-norescale} can be solved in the following two steps:
    \begin{enumerate}
        \item Solve the convex optimization problem \eqref{eq:convex} to get the optimal solution $z_{ij}$;
        \item Get the optimal $k_{ij}=s_{ij}^{-1}(\frac{1}{z_{ij}})$ where $s_{ij}^{-1}$ is the inverse of the speedup function $s_{ij}$.
    \end{enumerate}
    \label{thm:convex}
\end{theorem}
\begin{proof}

Define $z_{ij}:=\frac{1}{s_{ij}(k_{ij})}$. Define $d_{ij}:=\sup_{k\geq \base_i} s_{ij}(k)$ (if it does not exist, let $d_{ij}=\infty$). If there exists $k$ such that $s_{ij}(k)=d_{ij}$, denote the smallest one by $\xi_{ij}$; otherwise, let $\xi_{ij}$ be $\perp$.

Note that since $s_{ij}$ is non-decreasing and concave, we have that it is strictly increasing in $[\base_i,\xi_{ij}]$ ($[\base_i,\infty)$ if $\xi_{ij}=\perp$).
Thus we can define the function $\beta_{ij} = s_{ij}^{-1}$ to be the inverse of the speedup function on $[s_{ij}(\base_i),d_{ij}]$ ($[s_{ij}(\base_i), d_{ij})$ if $\xi_{ij}=\perp$).
Using the fact that $s_{ij}$ is increasing, positive and concave, we have that $\beta_{ij}$ is increasing and convex.

By definition of $z_{ij}$, we have that $k_{ij}=\beta_{ij}(\frac{1}{z_{ij}})$. Substituting this into the optimization problem \eqref{eq:online-opt-norescale}, we have 
\begin{equation}
    \begin{aligned}
    & \underset{z_{ij}} {\min}
    & & \sum_{i,j}\rho_{ij}z_{ij}\\
    & \text{subject to}
    & & \sum_{i,j} \rho_{ij} z_{ij}\beta_{ij}(\frac{1}{z_{ij}})\leq b\\
    & & & \frac{1}{s_{ij}(\base_i)}\geq z_{ij} \geq \frac{1}{d_{ij}}.\\
    \end{aligned}
    \label{eq:convex}
\end{equation}

Note that \[\frac{d}{dz_{ij}}(z_{ij}\beta_{ij}(\frac{1}{z_{ij}}))=\beta_{ij}(\frac{1}{z_{ij}}) - \frac{1}{z_{ij}}\beta_{ij}'(\frac{1}{z_{ij}}),\]
and 
\begin{align*}
    \frac{d^2}{dz_{ij}^2}(z_{ij}\beta_{ij}(\frac{1}{z_{ij}}))&=\frac{-1}{z_{ij}^2}\beta_{ij}'(\frac{1}{z_{ij}}) +\frac{1}{z_{ij}^2}\beta_{ij}'(\frac{1}{z_{ij}}) \\
    &\quad + \frac{1}{z_{ij}^3}\beta_{ij}''(\frac{1}{z_{ij}})\\
    &=\frac{1}{z_{ij}^3}\beta_{ij}''(\frac{1}{z_{ij}})>0 .
\end{align*}
Thus the optimization problem (\ref{eq:convex}) is a convex optimization, which we can numerically solve for the optimal $z_{ij}$.

Finally, given the optimal $z_{ij}$, we can get the optimal $k_{ij}$ by letting $k_{ij}=\beta_{ij}(\frac{1}{z_{ij}})$.

\end{proof}

\section{Detailed heuristic for BOA width calculator}
\label{app:boa width}

We provide the detailed pseudocode for our BOA width calculator as follows (Algorithm~\ref{alg:boa_width_calculator}). We note that we only implement the glue-parameter-based heuristic for homogeneous clusters (i.e., a heuristic for \eqref{eq:hetero-opt-rescale} with a single GPU type). For the BOA policy on heterogeneous clusters, we directly use a solver to solve the MICP \eqref{eq:hetero-opt-rescale}.
\begin{algorithm}[h]
\caption{BOA Width Calculator}
\label{alg:boa_width_calculator}
\begin{algorithmic}[1]
\STATE {\bfseries Input:} Total budget $b$, Job classes $1,\dots,M$ with epoch counts $l_1,\dots,l_M$.
\STATE {\bfseries Output:} Optimal parameters $\{k_{ij}^*\}$ for the fixed-width policy.
\STATE \COMMENT{First Step: Generate a set of glue configurations.}
\STATE $\mathcal{G} \leftarrow \emptyset$
\FOR{$n \leftarrow 1$ {\bfseries to} $50$}
    \FORALL{job classes $i \in \{1,\dots,M\}$}
        \STATE Define candidate set $S_i = \{2^0, 2^1, \dots, 2^{\lfloor \log_2 l_i \rfloor}\}$
        \STATE Sample $g_i$ uniformly from $S_i$
    \ENDFOR
    \STATE Add configuration $G = \{g_1, \dots, g_M\}$ to $\mathcal{G}$
\ENDFOR
\STATE \COMMENT{Second Step: Solve a feasible solution for each glue configuration.}
\STATE $\mathbb{E}[T]_{\min} \leftarrow \infty$
\STATE $\{k_{ij}^*\} \leftarrow \emptyset$
\FORALL{$G \in \mathcal{G}$}
    \STATE Construct super-epochs by gluing adjacent epochs for each class $i$ according to $g_i \in G$
    \STATE $b_{run} \leftarrow b$
    \STATE $b' \leftarrow \infty$
    \WHILE{$b' > b$}
        \STATE Compute $\{k_{ij}\}$ by solving \textbf{Optimization Problem~\eqref{eq:online-opt-norescale}} with budget $b_{run}$
        \STATE Round each $k_{ij}$ to the nearest integer on the non-decreasing concave hull of $s_{ij}$
        \STATE Compute $\mathbb{E}[T]$ and total cost $b'$ using $\{k_{ij}\}$ (Definition~\ref{def:BOA-hetero-rescale} with a single GPU type)
        \IF{$b' > b$}
            \STATE $b_{run} \leftarrow 0.99 \cdot b_{run}$
        \ELSE
            \IF{$\mathbb{E}[T] < \mathbb{E}[T]_{\min}$}
                \STATE $\mathbb{E}[T]_{\min} \leftarrow \mathbb{E}[T]$
                \STATE $\{k_{ij}^*\} \leftarrow \{k_{ij}\}$
            \ENDIF
        \ENDIF
    \ENDWHILE
\ENDFOR
\STATE {\bfseries return} $\{k_{ij}^*\}$
\end{algorithmic}
\end{algorithm}




\section{Experiments on a three-class subset of newTrace}
\label{app:subtrace}

The two experiments in this appendix use a subset of \textbf{newTrace} that keeps three of its five job classes: ResNet18, DeepSpeech2, and BERT, which together make up ${\sim}96\%$ of its jobs.
The first experiment requires the ground-truth speedup function of every class and epoch, which had to be profiled separately, and profiling three classes keeps this tractable.
The second experiment uses the same subset so that the two are comparable.
Both run in the Pollux setting.

\subhead{Sensitivity to speedup prediction error.}
\begin{figure}[h]
    \centering
    \includegraphics[width=0.6\linewidth]{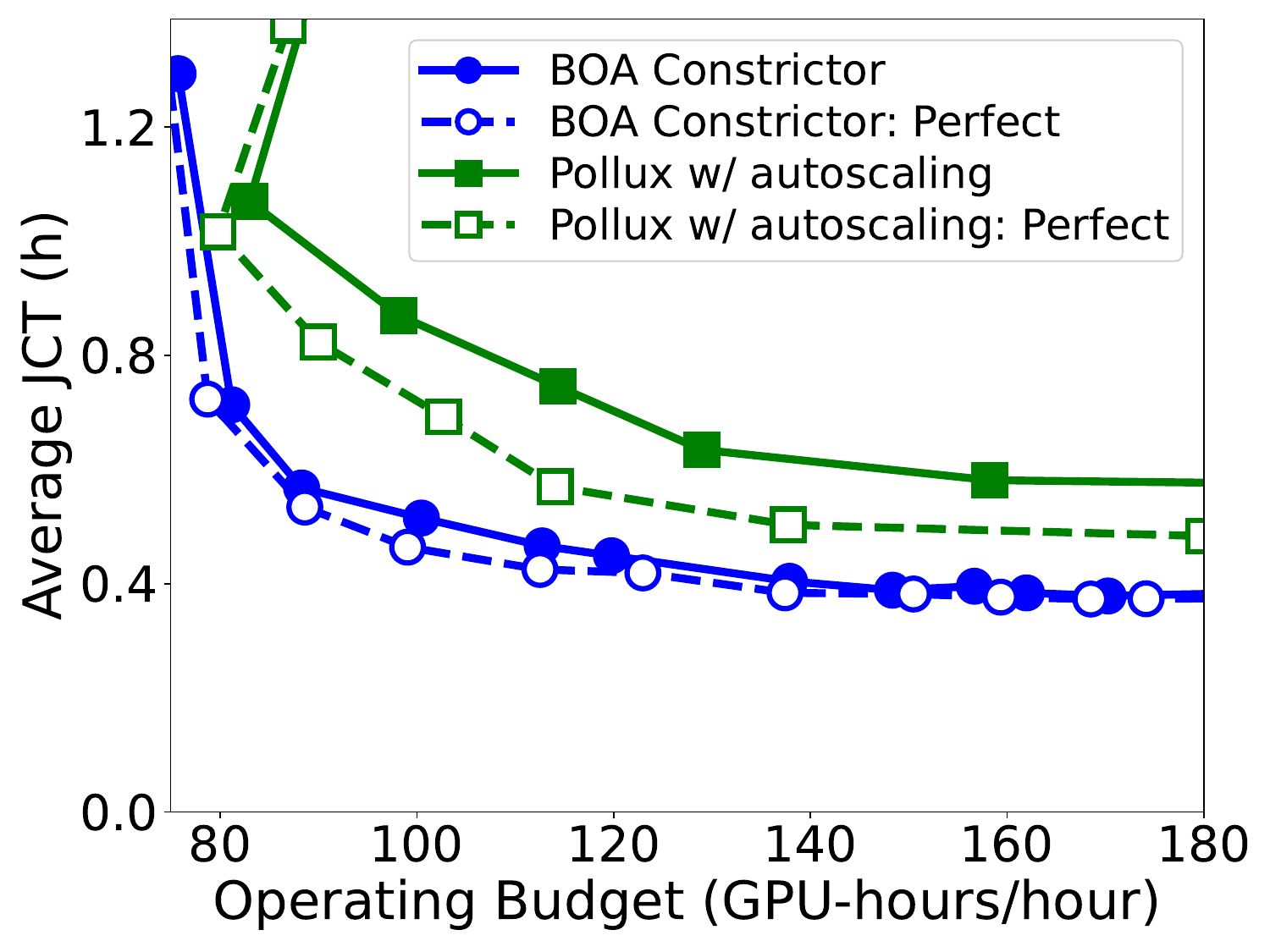}
    \caption{Average JCT for policies with perfect / imperfect speedup function information. \boa is less sensitive to modeling accuracy  than Pollux-with-Autoscaling.}
    \label{fig:sim:perfect}
\end{figure}
Even when the models used to predict job speedups have converged, these models are still imperfect.
We therefore evaluate \boa's sensitivity to errors in the speedup models.
Figure~\ref{fig:sim:perfect} compares each policy's performance when using an imperfect speedup model (as in Section~\ref{sec:eval:sim}) to its performance when given perfect profiling information about each job.
While the average JCT under \boa is essentially unchanged by prediction error, the average JCT under Pollux-with-Autoscaling increases by up to $1.4\times$ in the face of prediction errors.

\subhead{Starting without profiles or arrival rates.}
\begin{figure}[h]
    \centering
    \begin{subfigure}[b]{0.48\linewidth}
        \centering
        \includegraphics[width=\linewidth]{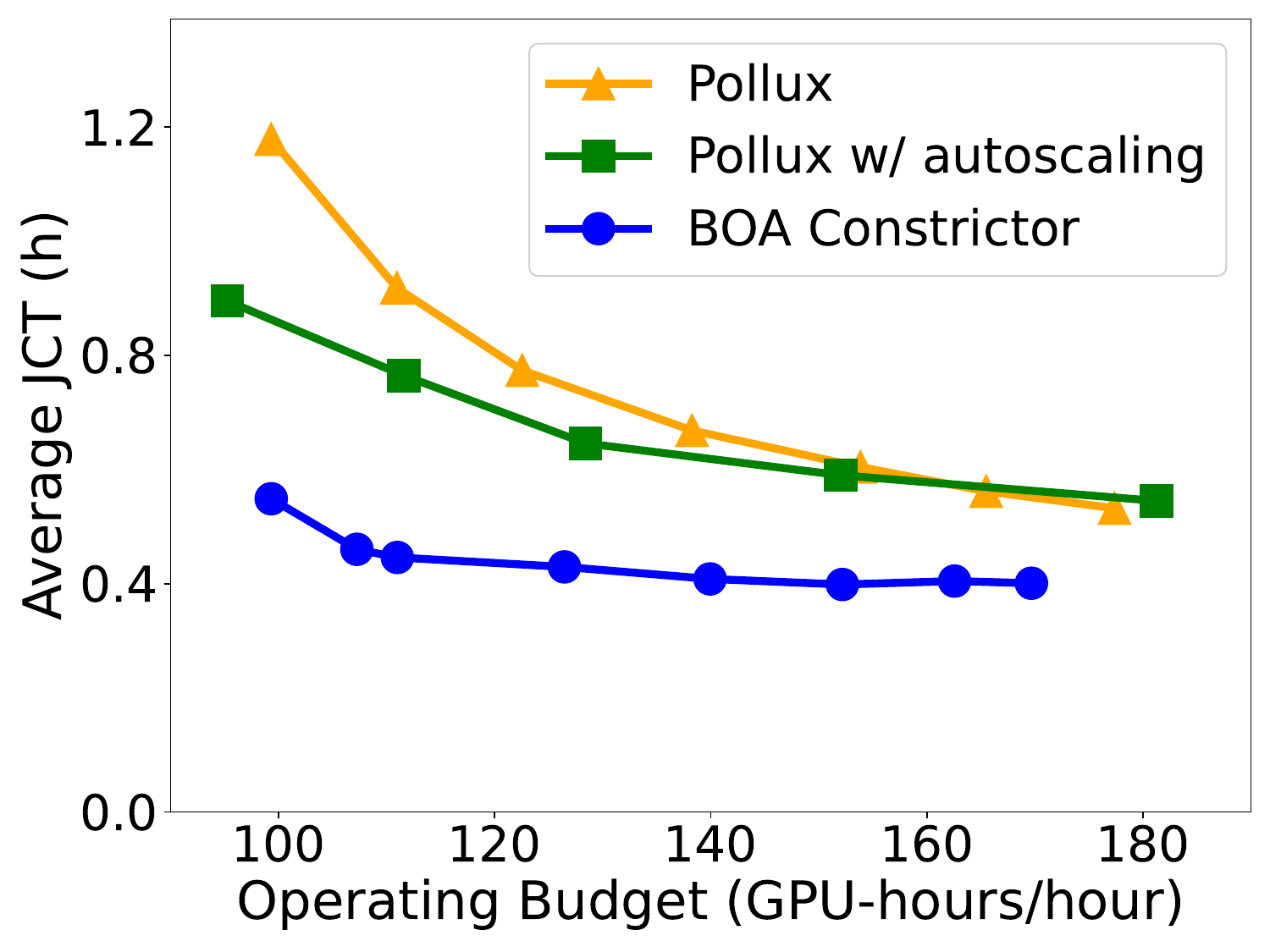}
        \caption{Average JCT.}
        \label{fig:coldstart:mean}
    \end{subfigure}
    \hfill
    \begin{subfigure}[b]{0.48\linewidth}
        \centering
        \includegraphics[width=\linewidth]{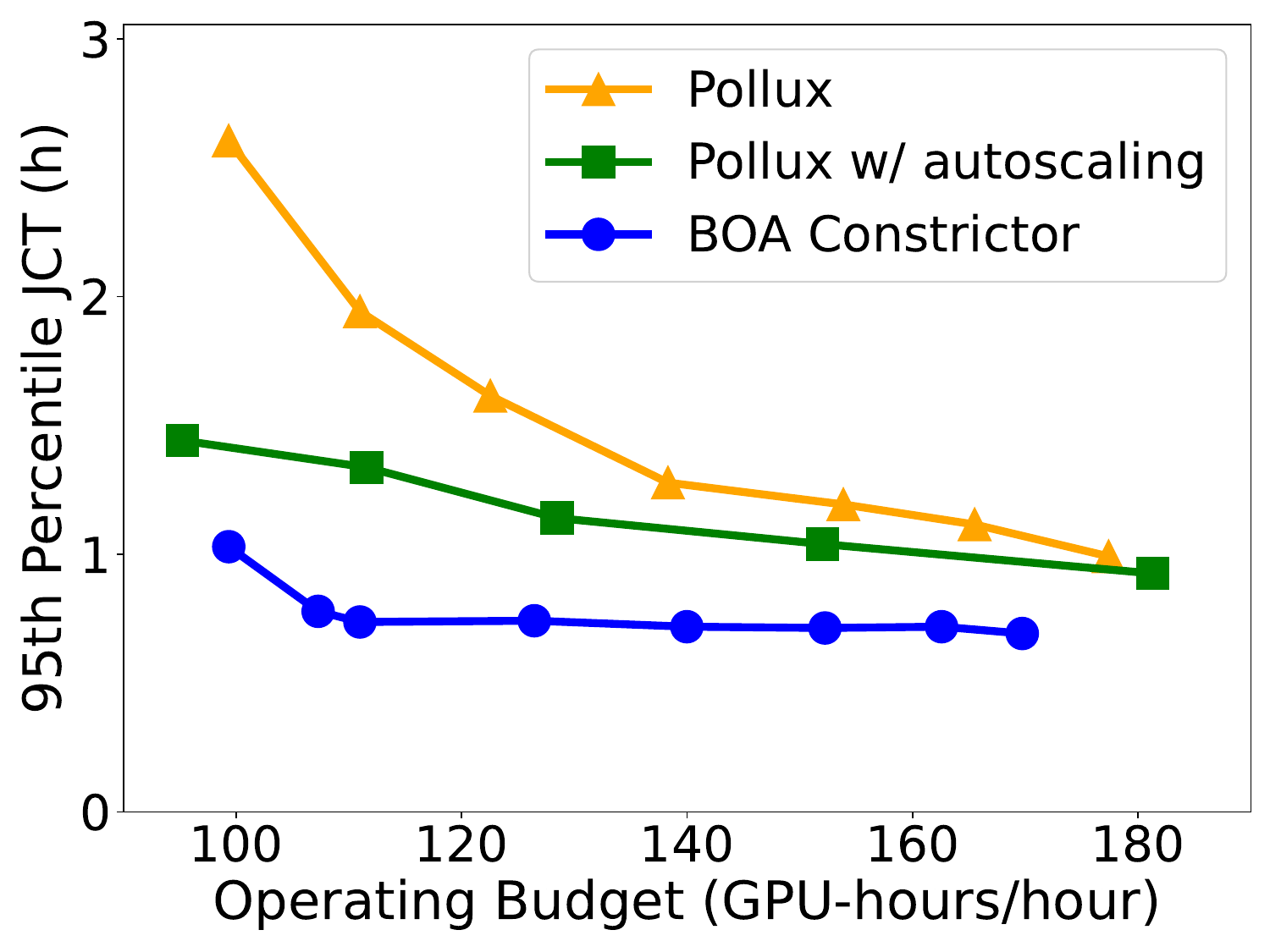}
        \caption{P95 JCT.}
        \label{fig:coldstart:tail}
    \end{subfigure}
    \caption{Comparing \boa with its competitors on the three-class subset of \textbf{newTrace} when every policy learns its performance model from scratch and \boa estimates arrival rates from \textbf{newTrace2}. \boa improves the Pareto frontier of average JCT and P95 JCT as much as it does in Figure~\ref{fig:sim}(a, d).}
    \label{fig:coldstart}
\end{figure}

This experiment (Section~\ref{sec:eval:sensitivity}) starts every policy without a performance model and \boa without the arrival rates of the trace.
It runs the three-class subset with two changes relative to the main experiments.

\subhead{The performance model is learned from scratch.}
The performance model of the Sia simulator predicts each class's speedup function and its inherent work per epoch.
In our main experiments, the model has been fit before the experiment starts, and \boa reads both $s_{ij}$ and $\E[X_{ij}]$ from it.
Here, the model has no data when the experiment begins.
Every policy, including \boa, uses the model as it is learned online from the throughput and gradient statistics of the jobs that have run so far.

\subhead{Arrival rates are estimated from a history trace.}
In our main experiments, $\lambda_i$ is computed from the trace being run.
Here, \boa estimates $\lambda_i$ from \textbf{newTrace2}~\cite{subramanya2023sia}, a second trace sampled from the same production system as \textbf{newTrace}, which plays the role of the customer's job history.
The estimated rates differ from those of \textbf{newTrace} by about $10\%$ overall.
During the experiment, \boa keeps a running estimate of each $\lambda_i$ from the arrivals observed so far, starting from the history-trace estimate.
The \boa Width Calculator uses the current estimates and the current performance model each time it re-solves for the widths.
The competitor policies do not use arrival rates and are unaffected by this change.

It turns out that with these two changes, \boa still retains almost the same gain over the competitors. This shows that the workload statistics can be efficiently learned and utilized with a reasonably long trace.


\section{Loss of the arrival-routing simplification}
\label{app:simplification}

Section~\ref{sec:theory:rescaling} restricts attention to policies that route each job to one GPU type on arrival (Definition~\ref{def:BOA-hetero-rescale}). This appendix compares the optimal hetero-fixed-width policy to the optimal unrestricted per-epoch policy in the no-rescaling setting, showing that the simplification incurs negligible loss in average JCT.

The unrestricted optimization, parameterized by per-epoch routing $\{p_{ij}^{(h)}\}$ and widths $\{k_{ij}^{(h)}\}$, minimizes the average JCT subject to the operating cost being at most $b$:
\begin{equation}
\begin{aligned}
& \underset{\{k_{ij}^{(h)},p_{ij}^{(h)}\}}{\min}
& & \frac{1}{\lambda}\sum_{i,j,h}\frac{p_{ij}^{(h)}\,\rho_{ij}}{s_{ij}^{(h)}(k_{ij}^{(h)})}\\
& \text{subject to}
& & \sum_{i,j,h}\frac{c^{(h)}\,p_{ij}^{(h)}\,\rho_{ij}\,k_{ij}^{(h)}}{s_{ij}^{(h)}(k_{ij}^{(h)})}\leq b,\\
& & & \sum_{h\in H}p_{ij}^{(h)}=1\quad\forall i,j,\\
& & & k_{ij}^{(h)}\geq \base_i^{(h)},\ p_{ij}^{(h)}\in[0,1].
\end{aligned}
\label{eq:hetero-opt-norescale}
\end{equation}

The hetero-fixed-width optimization replaces $p_{ij}^{(h)}$ with per-class arrival routing $p_i^{(h)}$ (no $j$ index):
\begin{equation}
\begin{aligned}
& \underset{\{k_{ij}^{(h)},p_i^{(h)}\}}{\min}
& & \frac{1}{\lambda}\sum_{i,j,h}\frac{p_i^{(h)}\,\rho_{ij}}{s_{ij}^{(h)}(k_{ij}^{(h)})}\\
& \text{subject to}
& & \sum_{i,j,h}\frac{c^{(h)}\,p_i^{(h)}\,\rho_{ij}\,k_{ij}^{(h)}}{s_{ij}^{(h)}(k_{ij}^{(h)})}\leq b,\\
& & & \sum_{h\in H}p_i^{(h)}=1\quad\forall i,\\
& & & k_{ij}^{(h)}\geq \base_i^{(h)},\ p_i^{(h)}\in[0,1].
\end{aligned}
\label{eq:hetero-fw-opt-norescale}
\end{equation}

We compared the average JCT of the two above approaches at a variety of budget levels, and our simplification was always within 5\% of the true optimal average JCT.


\section{Experimental setup details}
\label{app:workloads}
\label{app:setup}
This appendix gives the details of the workloads, cluster hardware, and GPU prices summarized in Section~\ref{sec:eval:setup}.

\subhead{Workloads.}
Both simulation workloads are sampled from production traces.
\textbf{newTrace}~\cite{subramanya2023sia}, used in the Pollux and Sia settings, samples $960$ jobs over $48$ hours from a production DL training system spanning multiple clusters with thousands of GPUs.
The \textbf{Rubick} workload~\cite{zhang2025rubick}, used in the Rubick setting, down-samples the busiest $12$ hours of the Microsoft Philly trace~\cite{jeon2019analysis} to $406$ jobs; its jobs train one of seven Transformer-based models, two of which (LLaMA-2-7B and LLaMA-30B) cannot fit on a single GPU (i.e., $\base_i>1$).
The implementation experiment (Section~\ref{sec:eval:imple}) uses a subsample of the \textbf{workload-1} trace from~\cite{qiao2021pollux}, $85$ ResNet18, BERT, and DeepSpeech2 jobs, chosen to keep AWS costs manageable.
Table~\ref{tbl:workload-rubick} lists the models, datasets, and job mix of the two simulation workloads.

{\small
\begin{table}[h]
    \centering
    \begin{tabular}{|c|c|c|}
        \hline
        \textbf{Model} & \textbf{Dataset} & \textbf{Share of jobs} \\
        \hline\hline
        \multicolumn{3}{|c|}{\textbf{newTrace}} \\ \hline
        ResNet18    & CIFAR-10    & 50.42\% \\ \hline
        DeepSpeech2 & CMU-ARCTIC  & 23.54\% \\ \hline
        BERT        & SQuAD       & 21.67\% \\ \hline
        YOLOv3      & PASCAL-VOC  & 4.75\%  \\ \hline
        ResNet50    & ImageNet    & 0.62\%  \\
        \hline\hline
        \multicolumn{3}{|c|}{\textbf{Rubick}} \\ \hline
        GPT-2       & Wikipedia    & 17.00\% \\ \hline
        BERT        & Wikipedia    & 15.76\% \\ \hline
        T5          & Wikipedia    & 15.27\% \\ \hline
        RoBERTa     & WikiText-2   & 14.29\% \\ \hline
        LLaMA-2-7B  & WuDaoCorpora & 13.30\% \\ \hline
        LLaMA-30B   & WuDaoCorpora & 12.56\% \\ \hline
        ViT         & ImageNet-1K  & 11.82\% \\ \hline
    \end{tabular}
    \caption{Composition of the \textbf{newTrace} workload~\cite{subramanya2023sia} ($960$ jobs total; Pollux and Sia settings) and the \textbf{Rubick} workload~\cite{zhang2025rubick} ($406$ jobs total; Rubick setting).}
    \label{tbl:workload-rubick}
\end{table}
}

\subhead{Cluster hardware.}
For each setting we configure our experiments to match the GPU types used in the corresponding paper.
The Pollux cluster is composed of NVIDIA T4 GPUs; our AWS implementation (Section~\ref{sec:eval:imple}) uses the same GPU type via \texttt{g4dn.12xlarge} instances ($4\times$~T4 per node).
The Rubick cluster is composed of NVIDIA A800 GPUs ($8\times$ A800 per node, $100$~Gb/s RDMA).
The original Sia cluster used NVIDIA T4, RTX~2080Ti, and V100 GPUs; no major hyperscaler offers the RTX~2080Ti, so we replace it by the A100, using the A100's speedup profiles and price.
Because the two homogeneous settings use a single GPU type, we report their cost in GPU-hours; in the heterogeneous setting, we report cost in USD per hour.

\label{app:heter-pricing}

\subhead{GPU prices.}
Table~\ref{tbl:heter-pricing} lists the Azure SKUs and on-demand list prices we use to compute USD-denominated operating budget in the heterogeneous Sia setting (Section~\ref{sec:eval:setup}).

{\small
\begin{table}[h]
    \centering
    \setlength{\tabcolsep}{4pt}
    \begin{tabular}{|l|c|c|c|}
        \hline
        \textbf{GPU type} & \textbf{Memory} & \textbf{Azure SKU} & \makecell{\textbf{USD per}\\\textbf{GPU-hour}} \\
        \hline\hline
        T4   & 16~GB & NC64as\_T4\_v3 & 1.088 \\ \hline
        V100 & 32~GB & ND40rs\_v2     & 2.754 \\ \hline
        A100 & 40~GB & ND96asr\_v4    & 3.400 \\ \hline
    \end{tabular}
    \caption{On-demand prices used to compute USD-denominated cost in the heterogeneous Sia setting (Azure list prices, \texttt{eastus}, queried May~2026).}
    \label{tbl:heter-pricing}
\end{table}
}

\section{Implementing the autoscaling variants}
\label{app:autoscaling-extensions}

This appendix gives the full details of the three autoscaling baselines (Pollux-with-Autoscaling, Rubick-with-Autoscaling, and Sia-with-Autoscaling) summarized in Section~\ref{sec:eval:setup}. All three are built around the same cluster-efficiency heuristic that the corresponding fixed-cluster schedulers already use; the differences are in how that heuristic is operationalized when the cluster size becomes a decision variable.

The target-efficiency rule rests on the observation that allocating many GPUs to a job is inefficient, so cluster efficiency drops when few jobs are in the system; the scheduler therefore sets a target efficiency and grows or shrinks the cluster to meet it.
AdaptDL ships only a prototype of this rule, which was only evaluated on a single training job in \cite{qiao2021pollux}.
This prototype does not scale to realistic workloads: it optimizes the cluster size jointly with all job allocations and hyperparameters, and as the number of jobs and the cluster size grow, this joint combinatorial optimization either takes too long or fails to find a valid solution (as also observed in \cite{subramanya2025copter}).
We therefore reimplemented the rule on AdaptDL, a significant engineering effort; we call the result Pollux-with-Autoscaling.
Following the simulations in~\cite{qiao2021pollux}, our reimplementation breaks the joint optimization into smaller pieces, searching over the cluster size separately from the job allocations.
The same mechanism applies to Rubick and Sia, which also maximize cluster efficiency, giving Rubick-with-Autoscaling and Sia-with-Autoscaling.

Although we built these variants ourselves, we tuned their parameters and design choices to the best of our ability.
Where the Pollux prototype fixes a design choice, we start from its setting and then tune it further; Section~\ref{sec:eval} evaluates the best-performing version of each variant.
The subsections below describe each choice and the alternatives we tried.

\subsection{Cluster efficiency}
The cluster efficiency at any time is the sum of all running jobs' \emph{normalized speedups} divided by the cluster size, where the normalized speedup of a class-$i$ job in epoch $j$ on $k$ GPUs is
\[
\tilde s_{ij}(k) \;=\; \frac{s_{ij}(k)}{s_{ij}(\base_i)}.
\]
In the heterogeneous setting, we apply the same definition per GPU type, with $\base_i^{(h)}$ and $s_{ij}^{(h)}$ replacing $\base_i$  and $s_{ij}$.

\subsection{Tuning the target efficiency $c$}

While~\cite{qiao2021pollux} only suggests setting the target efficiency at $0.5$, we explore the full range of the target efficiency level, $c$; this is a performance knob that we can tune to evaluate the autoscaling variants at different operating budgets. Specifically, we define $\Delta = \min\{0.3(1-c), 0.3c\}$. The cluster size is increased whenever the cluster efficiency exceeds $c + \Delta$, and decreased whenever the cluster efficiency falls below $c - \Delta$. The new cluster size is then set using combinatorial optimization to select the cluster size and job allocations whose cluster efficiency is closest to $c$. Importantly, although adjusting the target efficiency $c$ controls the GPU usage of the policy, there is no analysis predicting the average number of GPU-hours that an autoscaling variant will consume for a given value of $c$.

\subsection{Trigger band and scheduling quantum}

The band $[c-\Delta, c+\Delta]$ decides when autoscaling is triggered: the cluster is resized only when its efficiency leaves the band, so small fluctuations in efficiency do not resize the cluster.
Our choice of $\Delta$ generalizes the AdaptDL prototype~\cite{adaptdl}, which pairs the target $c=0.5$ with the band $[0.35, 0.65]$, i.e., $\Delta = 0.15 = 0.3c$.
We also tried a band of constant width, $\Delta = 0.15$ for every $c$; it produced almost the same cost-performance curve as our choice.
 
Pollux-with-Autoscaling makes its scheduling and autoscaling decisions every $60$ seconds.
This scheduling quantum is not our choice: it is the default of Pollux~\cite{qiao2021pollux,adaptdl}.
We also tried quanta of $30$ and $5$ seconds.
A $30$-second quantum performs about the same as the $60$-second default, while a $5$-second quantum clearly degrades performance.

\subsection{Rubick: non-monotone cluster efficiency}

In the Pollux setting, where every job fits on a single GPU, the cluster efficiency is monotone in the cluster size, and the cluster size matching a target efficiency $c$ can be located by binary search. In the Rubick setting, however, this monotonicity fails. For example, suppose the workload contains one ordinary job and one $\base_i=8$ job. At cluster size $4$, the large job is queued and the cluster efficiency is determined entirely by the ordinary job; at cluster size $12$, both jobs are running and the cluster efficiency becomes higher because the normalized speedup of the $\base_i=8$ job is exactly $1$. Because the cluster efficiency is no longer monotone in cluster size, the autoscaler cannot use binary search and must instead enumerate every feasible cluster size to find the one whose efficiency is closest to $c$. We tolerate this overhead in our simulator (the simulator pauses while the sweep completes) for the sake of a fair comparison; even so, \boa outperforms Rubick-with-Autoscaling (Section~\ref{sec:eval:sim}).

\subsection{Sia: which cluster to scale}

In the heterogeneous Sia setting, when the cluster efficiency drifts away from its target, the heuristic does not specify which GPU-type cluster should be grown or shrunk; the adaptive-training literature on heterogeneous clusters provides no answer either. Our implementation maintains the same efficiency target for each GPU-type cluster and lets each cluster autoscale accordingly and independently.

We also tried the alternative: a single efficiency target for the whole heterogeneous cluster, scaling every GPU-type cluster by the same proportion whenever the overall cluster efficiency leaves its band.
This alternative performed worse for two reasons.
First, it greatly increases job rescaling, because every scaling decision resizes all GPU-type clusters at once.
Second, proportional scaling is too coarse when a GPU-type cluster has few GPUs; for example, a cluster of $4$ GPUs cannot change by less than $25\%$.
Even with the better of the two designs, \boa still outperforms Sia-with-Autoscaling (Section~\ref{sec:eval:sim}).

\section{Extending the AdaptDL Framework}
\label{app:adaptdl-extensions}

We extend the AdaptDL framework to help evaluate \boa and improve user experience.


\subhead{Profile Sharing.}
The original AdaptDL framework trains a performance model from scratch for every new job. While not central to our implementation, we extended the system to allow jobs of the same class to share performance models. This is realized by a global profiler on top of the AdaptDL schedulers. The main purpose of this modification is to isolate our improvement in scheduling from the orthogonal hyperparameter tuning process (see Section~\ref{sec:eval:imple}).

\subhead{Decision Support Tool.}
We leverage the predictability of the BOA policy to build an offline profiling calculator. Given a workload trace as input, \boa can provide the full Pareto frontier by running Algorithm~\ref{alg:boa_width_calculator} at different budgets. This tool allows a customer to select an optimal budget $b$ based on their desired average JCT before provisioning any real resources.

\subhead{Extensibility to Hyperparameter Tuning.}
While we deploy the batch-size and learning-rate tuning used in~\cite{qiao2021pollux} for each job in our evaluations, the BOA policy is compatible with \emph{any} hyperparameter tuning approach. The BOA policy and the hyperparameter optimization solve orthogonal problems: the BOA policy optimizes system performance given the speedup functions, while the hyperparameter optimization improves those speedup functions. \boa does not need to know the particular hyperparameter choices, so long as it can obtain the speedup functions.

\end{document}